\documentclass[a4paper,11pt]{article}
\usepackage[left=3.17cm,right=3.17cm,top=2.54cm,
headheight=0.5cm,headsep=0.54cm,bottom=2.54cm,footskip=0.79cm
]{geometry}

\usepackage{graphicx,hyperref}

\usepackage{amsmath,amssymb,amsthm,cite}

\newtheorem{corollary}{Corollary}
\newtheorem{definition}{Definition}
\newtheorem{lemma}{Lemma}
\newtheorem{proposition}{Proposition}
\newtheorem{theorem}{Theorem}

\begin{document}

\title{Connes spectral distance on twisted fuzzy torus}
	
\author{Zhi-Kang You$^1$, Bing-Sheng Lin$^{1,2}$\thanks{Corresponding author. 
	E-mail address: sclbs@scut.edu.cn}\\
	\small $^{1}$ School of Mathematics, South China University of Technology,
	Guangzhou 510641, China\\
	\small $^{2}$ Institute for Mathematics, Astrophysics and Particle Physics,\\
	\small  Radboud University Nijmegen, Heyendaalseweg 135, 6525 AJ Nijmegen, The Netherlands.
}

\date{\today}

\maketitle
	
\begin{abstract}
In this paper, we study the Connes spectral distance between states on the fuzzy torus. We construct a Dirac operator by commutators and anticommutators. Based on this Dirac operator, we construct a spectral triple of the fuzzy torus.
We study some properties of the spectral distance on the fuzzy torus. We find that there is a reciprocal Pythagorean theorem between the spectral distances.
We construct a conditional expectation function of the optimal element which can lead to a contraction of the corresponding Lipschitz seminorm. We find that for any diagonal states, the corresponding optimal elements of spectral distances are also diagonal.
We explicitly calculate the spectral distances of some simple states, including basic states and some simple mixed states. We find that there are some kinds of cyclic symmetry in both the optimal elements and the spectral distances between the diagonal states.
Furthermore, we also construct a fuzzy torus with some type of conformal twist, and study the relation between conformal parameters and spectral distances.
\end{abstract}


\section{Introduction}
One of the central ideas of noncommutative geometry is to replace the classical description of a space by points and local coordinates with algebraic and operator-theoretic data encoded in a spectral triple $(\mathcal{A},\mathcal{H},D)$ \cite{Connes}. Here, $\mathcal{A}$ is an involutive algebra, $\mathcal{H}$ is a Hilbert space carrying a representation of $\mathcal{A}$, and $D$ is the Dirac operator, which determines the differential and metric structures of the space. Due to the noncommutativity, there is no traditional point in a spectral triple.
In a spectral triple, a pure state is the
analog of a traditional point in a normal commutative space.
There is no normal concept of distance between two points in a noncommutative space, but
one can calculate some kinds of distance measure between the states, such as the
Connes distance \cite{Connes1}.
Connes spectral distance
defines a metric on the state space purely in terms of the algebra and the Dirac operator. In the commutative case, it recovers the geodesic distance. In the noncommutative setting, it has become a useful tool for investigating the geometry of quantum state spaces.
The Connes spectral distances in some kinds
of noncommutative spaces have already been studied in literature
\cite{Bimonte,Cagnache,Wallet,Martinetti,DAndrea,Pythagoras,Franco,Chaoba,Revisiting,Kumar,Chakraborty,Lin1,Lin2,Clare,Kumar1}.

The introduction of a real structure $J$ further enriches spectral triples by naturally incorporating gauge symmetries\cite{Connes2,cacic,Chakraborty1}. For a real spectral triple $(\mathcal{A},\mathcal{H},D,J)$, the operator $J$ induces a right action compatible with the left representation, allowing unitary elements of the algebra to lift to gauge transformations on the Hilbert space. The inner fluctuation of the Dirac operator
produces gauge potentials analogous to classical connections.
Despite these developments, the reality condition and the first-order condition in standard real spectral triples are often too restrictive to accommodate natural geometric deformations. Motivated by these issues, twisted real structures have recently attracted considerable attention\cite{Filaci,Brzezinski1,Magee,Dabrowski,Martinetti2,Dabrowski1,Magee1}.

Fuzzy spaces are important types of noncommutative geometries\cite{Madore,Barrett1,Kimura}. Among them, the fuzzy torus is one of the most representative examples\cite{Bigatti,Bigatti1,Latremoliere,Perez}. It can be regarded as a finite-dimensional approximation of both the classical two-dimensional torus and the quantum torus.
In this paper, we will study spectral distances of states on the fuzzy torus with and without twisted real structure.

This paper is organized as follows. In Sec.~\ref{sec2}, we review the construction of the spectral triple of the fuzzy torus, and we propose a new construction of the Dirac operator that includes both commutators and anticommutators, as well as two parameters. 
In Sec.~\ref{sec3}, we study Connes spectral distances on the fuzzy torus. we derive some elementary properties of Connes spectral distances and optimal elements. We find that the spectral distances of some fuzzy tori satisfy a reciprocal Pythagorean theorem.
We introduce a condition expectation function which can lead to a contraction of the corresponding Lipschitz seminorm.
We find that the optimal elements for the diagonal states are still diagonal.
In Sec.~\ref{sec4}, we explicitly calculate the Connes spectral distances between some simple states.
We find that there are some kinds of cyclic symmetry in the spectral distances between diagonal states.
Furthermore, in Sec.~\ref{sec5}, we also study the spectral distances of states on the fuzzy torus with twisted real structure. We find that the spectral distance in the real spectral
triple with conformal twist depend on the conformal parameters.
Some discussions and conclusions are given in Sec.~\ref{sec6}.

\section{Fuzzy Torus}\label{sec2}

First, let us review the spectral triple of the fuzzy torus given in Ref.~\cite{Bigatti1}.

\begin{definition}
  \cite{Bigatti1} A noncommutative torus is a pair of unitary operators $U, V$ on a
  Hilbert space $\mathfrak{h}$ together with a complex number $q$, satisfying
  the relation
  \begin{equation}
    UV = qVU.
  \end{equation}
\end{definition}

When $\mathfrak{h}$ is finite-dimensional, the corresponding noncommutative
torus is called finite. Using the unitarity condition, we have $|q| = 1$.

\begin{definition}
  \cite{Bigatti1} Let $(U, V, \mathfrak{h})$ be a finite noncommutative torus. The
  torus algebra $\langle U, V \rangle$ is the $\ast$-subalgebra of
  $\mathrm{End} (\mathfrak{h})$ generated by $U, V$.
\end{definition}

\begin{definition}
  \cite{Bigatti1} A real structure on a noncommutative torus $(U, V, \mathfrak{h})$
  is an antiunitary map $\mathfrak{j}: \mathfrak{h} \rightarrow \mathfrak{h}$,
  which is also an involution, $\mathfrak{j}^2 = 1$, and satisfies
  \begin{equation}
    [\mathfrak{j}U\mathfrak{j}, U] = 0, \quad [\mathfrak{j}U\mathfrak{j}, V] =
    0, \quad [\mathfrak{j}V\mathfrak{j}, U] = 0, \quad
    [\mathfrak{j}V\mathfrak{j}, V] = 0.
  \end{equation}
\end{definition}

The real structure defined above also determines a right action of $U, V$ on
$\psi \in \mathfrak{h}$:
\begin{equation}
  \psi U := \mathfrak{j}U^{\ast} \mathfrak{j} \psi, \qquad \psi V :=
  \mathfrak{j}V^{\ast} \mathfrak{j} \psi .
\end{equation}
These right actions are also unitary,
\begin{equation}
  \psi U^{\ast} = \psi U^{- 1}, \qquad \psi V^{\ast} = \psi V^{- 1} .
\end{equation}
\begin{definition}
  \cite{Bigatti1} A noncommutative torus with a real structure $(U, V,
  \mathfrak{h}, \mathfrak{j})$ is called a fuzzy torus.
\end{definition}

\begin{definition}
  \cite{Bigatti1} A finite real spectral triple is defined as:
  \begin{itemize}
    \item A finite-dimensional Hilbert space $\mathcal{H}$.
    
    \item A $\ast$-algebra $A$ of operators with a faithful left action in
    $\mathcal{H}$.
    
    \item An integer $d$ modulo $8$, called the KO-dimension.
    
    \item For even KO-dimension $d$, there exists an Hermitian operator $\Gamma
    : \mathcal{H} \rightarrow \mathcal{H}$ called the chirality operator,
    which commutes with the algebra $A$, $a \Gamma = \Gamma a$, and satisfies
    $\Gamma^2 = 1$.
    
    \item An antiunitary map $J : \mathcal{H} \rightarrow \mathcal{H}$ called
    the real structure, satisfying
    \begin{equation}
      [a, JbJ^{- 1}] = 0, \quad \forall a, b \in A,
    \end{equation}
    and $J^2 = \varepsilon$. For even KO-dimension $d$, we also have $J \Gamma
    = \varepsilon'' \Gamma J$.
    
    \item An Hermitian operator $D : \mathcal{H} \rightarrow \mathcal{H}$
    called the Dirac operator, satisfying
    \begin{equation}
      [[D, a], JbJ^{- 1}] = 0, \quad \forall a, b \in A,
    \end{equation}
    and $DJ = \varepsilon' JD$. For even KO-dimension $d$, we also have $D
    \Gamma + \Gamma D = 0$.
  \end{itemize}
\end{definition}

The above symbols $\varepsilon, \varepsilon', \varepsilon''$ are related to
the KO-dimension $d$ as follows:
\[ \begin{array}{|c|c|c|c|c|c|c|c|c|}
     \hline
     d & 0 & 1 & 2 & 3 & 4 & 5 & 6 & 7\\
     \hline
     \varepsilon & 1 & 1 & - 1 & - 1 & - 1 & - 1 & 1 & 1\\
     \hline
     \varepsilon' & 1 & - 1 & 1 & 1 & 1 & - 1 & 1 & 1\\
     \hline
     \varepsilon'' & 1 &  & - 1 &  & 1 &  & - 1 & \\
     \hline
   \end{array} \]
Furthermore, we can also define a twisted structure on the real spectral
triple $(A, H, D, J)$.

\begin{definition}
  \cite{Magee,Dabrowski1} Let $\nu \in \mathrm{Aut} (H)$ be a self-adjoint
  invertible operator on the Hilbert space $H$, and assume that $A$ is
  invariant under the automorphism $\nu$, i.e., for any $a \in A$, $\nu^{- 1}
  a \nu \in A$. We say that the real spectral triple $(A, H, D, J)$ has a
  $\nu$-twisted real structure, if for all $a, b \in A$ the following
  conditions hold:
\begin{eqnarray}
    [D, a] J \nu^{- 2} b \nu^2 J^{- 1} &=& JbJ^{- 1}  [D, a],
\nonumber\\
    DJ \nu &=& \varepsilon' \nu JD,
\nonumber\\
    \nu J \nu &=& J,
\end{eqnarray}
  If $(A, H, D, J)$ is even-dimensional, with chirality operator $\Gamma : H
  \to H$, it additionally needs to satisfy:
  \begin{equation}
    \nu^2 \Gamma = \Gamma \nu^2 .
  \end{equation}
\end{definition}

In the following content, we will construct a spectral triple based on the fuzzy torus
$(U, V, \mathfrak{h}, \mathfrak{j})$ \cite{Bigatti1}. Consider the following
extended Hilbert space
\begin{equation}
  H =\mathbb{C}^4 \otimes \mathfrak{h},
\end{equation}
where $\mathbb{C}^4$ is the spinor space. The action of the algebra $A =
\langle U, V \rangle$ on $H$ is given by:
\begin{equation}
  U (v \otimes m) = v \otimes U (m), \quad V (v \otimes m) = v \otimes V (m),
  \quad v \in \mathbb{C}^4, \quad m \in \mathfrak{h}.
\end{equation}
Its real structure is defined as
\begin{equation}
  J (v \otimes m) = j (v) \otimes \mathfrak{j} (m),
\end{equation}
where $j$ is a real structure on the spinor space $\mathbb{C}^4$, satisfying
$j^2 = -\mathbb{I}_4$, and compatible with the Clifford generators: $j
\gamma^{\mu} = \gamma^{\mu} j$ ($\mu = 1, 2, 3, 4$), $\gamma^{\mu}$ are
anti-Hermitian Gamma matrices which satisfy the following relations,
\begin{equation}
  \{\gamma^{\mu}, \gamma^{\nu} \} = - 2 \delta_{\mu \nu} \mathbb{I}_4, \quad
  (\gamma^{\mu})^{\dag} = - \gamma^{\mu}.
\end{equation}
Furthermore, there is the Hermitian matrix $\gamma_5$,
\begin{equation}
  \gamma_5 = \gamma^1 \gamma^2 \gamma^3 \gamma^4, \quad \gamma_5^2
  =\mathbb{I}_4, \quad \{\gamma_5, \gamma^{\mu} \} = 0.
\end{equation}
The Gamma matrices can be explicitly expressed as:
\begin{eqnarray}
  \gamma^1 = \left( \begin{array}{cccc}
    0 & 0 & \mathrm{i} & 0\\
    0 & 0 & 0 & \mathrm{i}\\
    \mathrm{i} & 0 & 0 & 0\\
    0 & \mathrm{i} & 0 & 0
  \end{array} \right), && \gamma^2 = \left( \begin{array}{cccc}
    0 & 0 & - 1 & 0\\
    0 & 0 & 0 & 1\\
    1 & 0 & 0 & 0\\
    0 & - 1 & 0 & 0
  \end{array} \right),
\nonumber\\
  \gamma^3 = \left( \begin{array}{cccc}
    0 & 0 & 0 & \mathrm{i}\\
    0 & 0 & - \mathrm{i} & 0\\
    0 & - \mathrm{i} & 0 & 0\\
    \mathrm{i} & 0 & 0 & 0
  \end{array} \right), && \gamma^4 = \left( \begin{array}{cccc}
    0 & 0 & 0 & - 1\\
    0 & 0 & - 1 & 0\\
    0 & 1 & 0 & 0\\
    1 & 0 & 0 & 0
  \end{array} \right),
\end{eqnarray} 
and
\begin{equation}
  \gamma_5 = \left( \begin{array}{cccc}
    - 1 & 0 & 0 & 0\\
    0 & - 1 & 0 & 0\\
    0 & 0 & 1 & 0\\
    0 & 0 & 0 & 1
  \end{array} \right) .
\end{equation}
The real structure $j$ can be explicitly taken as the following anti-linear
operator:
\begin{equation}
  j = \left( \begin{array}{cccc}
    0 & 1 & 0 & 0\\
    - 1 & 0 & 0 & 0\\
    0 & 0 & 0 & - 1\\
    0 & 0 & 1 & 0
  \end{array} \right) \mathcal{C}, \qquad j \left( \begin{array}{c}
    v_1\\
    v_2\\
    v_3\\
    v_4
  \end{array} \right) = \left( \begin{array}{c}
    \overline{v_2}\\
    - \overline{v_1}\\
    - \overline{v_4}\\
    \overline{v_3}
  \end{array} \right),
\end{equation}
where $\mathcal{C}$ is the complex conjugation operator. The right action of
$A$ on $H$ induced by $J$ is
\begin{equation}
  (v \otimes m) a := Ja^{\ast} J^{- 1}  (v \otimes m) = v \otimes (ma) .
\end{equation}
The chirality grading operator $\Gamma$ is defined as
\begin{equation}
  \Gamma (v \otimes m) = (\gamma_5 v) \otimes m,
\end{equation}
where $\gamma_5$ is the chirality operator in the four-dimensional Clifford
algebra. Since $\Gamma$ acts only on the spinor part and the algebra $A$ acts
only on $\mathfrak{h}$, we have $[\Gamma, a] = 0$ for all $a \in A$.

In the following, using the construction in Ref. \cite{Bigatti1}, we will
specifically consider the following noncommutative torus $(C, S,  \mathbb{M}_N
(\mathbb{C}))$, where the $N \times N$ matrices $C$, $S$ are defined as
\begin{equation}\label{cs}
  C = \left( \begin{array}{ccccc}
    1 & 0 & 0 & \cdots & 0\\
    0 & q & 0 & \cdots & 0\\
    0 & 0 & q^2 & \cdots & 0\\
    \vdots & \vdots & \vdots & \ddots & \vdots\\
    0 & 0 & 0 & \cdots & q^{N - 1}
  \end{array} \right), \qquad S = \left( \begin{array}{ccccc}
    0 & 0 & \cdots & 0 & 1\\
    1 & 0 & \cdots & 0 & 0\\
    0 & 1 & \cdots & 0 & 0\\
    \vdots & \vdots & \ddots & \vdots & \vdots\\
    0 & 0 & \cdots & 1 & 0
  \end{array} \right),
\end{equation}
where $q = e^{2 \pi \mathrm{i} / N}$. In this case, the algebra $A = \langle C, S
\rangle \cong  \mathbb{M}_N (\mathbb{C})$, and
\begin{eqnarray}
&&  C^N = S^N =\mathbb{I}_N, \qquad C^{\dag} C = S^{\dag} S
  =\mathbb{I}, \qquad q^N = 1.
\nonumber\\
&&	CS = qSC, \quad C^{\dag}S = q^{-1}SC^{\dag}, \quad CS^{\dag} = q^{-1}S^{\dag}C, \quad C^{\dag}S^{\dag} = qS^{\dag}C^{\dag}.
\end{eqnarray}
Obviously, as $N \to \infty$, $q \to 1$, the noncommutative relation of $C, S$
gradually tends to a commutative one, so the noncommutative torus $(C, S,  \mathbb{M}_N
(\mathbb{C}))$ approaches the classical torus.

The inner product on the Hilbert space $\mathfrak{h}=  \mathbb{M}_N (\mathbb{C})$ is the
Hilbert-Schmidt inner product $(\psi, \phi) = \mathrm{Tr} (\psi^{\dag} \phi)$,
and the action of $C, S$ on $\mathfrak{h}=  \mathbb{M}_N (\mathbb{C})$ is realized by
ordinary matrix multiplication. The real structure on the noncommutative torus
$(C, S,  \mathbb{M}_N (\mathbb{C}))$ is simply the Hermitian conjugation of matrices on
$ \mathbb{M}_N (\mathbb{C})$, $\mathfrak{j}= (\cdot)^{\dag}$.

For simplicity, let us introduce the following generalized coordinate operators $X_{\mu} \in
A$ ($\mu = 1, 2, 3, 4$), satisfying the relations:
\begin{eqnarray}
  X_1 = \frac{1}{2}  (C^{\dag} + C), && X_2 = \frac{\mathrm{i}}{2} (C^{\dag} - C),
\nonumber\\
  X_3 = \frac{1}{2}  (S^{\dag} + S), && X_4 = \frac{\mathrm{i}}{2} (S^{\dag} - S) .
\end{eqnarray}
Hence there are
\begin{equation}
  X_{\mu}^{\dag} = X_{\mu}, \quad C = X_1 + \mathrm{i}X_2, \quad
  S = X_3 + \mathrm{i}X_4 .
\end{equation}
Similar to the construction in Ref.~\cite{Bigatti1}, the Dirac operator can be
defined as
\begin{eqnarray}
  D & = & \mathrm{i} \alpha \sum_{\mu = 1}^4 \gamma^{\mu} \otimes [X_{\mu},
    \cdot  ] + \beta \sum_{\mu = 1}^4 \gamma_5
  \gamma^{\mu} \otimes \{X_{\mu},   \cdot   \}
  \nonumber\\
  & = & \mathrm{i} \alpha \sum_{\mu = 1}^4 \gamma^{\mu} \otimes d_{\mu} + \beta
  \sum_{\mu = 1}^4 \gamma_5 \gamma^{\mu} \otimes \delta_{\mu} \nonumber\\
  & = & \alpha D_1 + \beta D_2 . 
\end{eqnarray}
where $\alpha, \beta \in \mathbb{R}$ are real parameters, and the operators $d_{\mu},\delta_{\mu}: A \rightarrow A$,
\begin{equation}
  d_{\mu} (a) := [X_{\mu}, a], \quad \delta_{\mu} (a) := \{X_{\mu}, a\},
  \qquad \mu = 1, 2, 3, 4, \quad a \in A.
\end{equation}
When $\alpha=1$, $\beta=0$, there is
\begin{equation}
  D=D_1 = \mathrm{i} \sum_{\mu = 1}^4 \gamma^{\mu} \otimes [X_{\mu},  
  \cdot  ] = \mathrm{i} \sum_{\mu = 1}^4 \gamma^{\mu} \otimes d_{\mu},
\end{equation}
and when $\alpha=0$, $\beta=1$, there is
\begin{equation}
  D=D_2 = \sum_{\mu = 1}^4 \gamma_5 \gamma^{\mu} \otimes \{X_{\mu},
    \cdot   \} = \sum_{\mu = 1}^4 \gamma_5
  \gamma^{\mu} \otimes \delta_{\mu} .
\end{equation}
It is easy to verify that the operators $d_{\mu}$, $\delta_{\mu}$ are all
Hermitian.
\begin{eqnarray}
  (a, d_{\mu} (b)) & = & (a, [X_{\mu}, b]) \nonumber\\
  & = & \mathrm{Tr} (a^{\dag} [X_{\mu}, b]) \nonumber\\
  & = & \mathrm{Tr} (a^{\dag} X_{\mu} b - a^{\dag} bX_{\mu}) \nonumber\\
  & = & \mathrm{Tr} (a^{\dag} X_{\mu}^{\dag} b - X_{\mu}^{\dag} a^{\dag} b)
  \nonumber\\
  & = & \mathrm{Tr} ([X_{\mu}, a]^{\dag} b) \nonumber\\
  & = & ([X_{\mu}, a], b) \nonumber\\
  & = & (d_{\mu} (a), b), 
\end{eqnarray}
Therefore the operator $d_{\mu}$ is Hermitian. Similarly, one can prove that
$\delta_{\mu}$ is also Hermitian.

Since the Gamma matrices $\gamma^{\mu}$ are anti-Hermitian and $\gamma_5$ is
Hermitian, we have $(\gamma_5 \gamma^{\mu})^{\dag} = (\gamma^{\mu})^{\dag}
(\gamma_5)^{\dag} = - \gamma^{\mu} \gamma_5 = \gamma_5 \gamma^{\mu}$. It is
easy to see that the above Dirac operators $D, D_1, D_2$ are all Hermitian.

\begin{proposition}
  $(A, H, D, J, \Gamma)$ is a finite real even spectral triple of
  KO-dimension $4$.
\end{proposition}

\begin{proof}
  For $\psi = v \otimes m, \phi = w \otimes n$, there is
  \begin{equation}
  	(J \psi, J \phi) = (j (v), j (w))_{\mathbb{C}^4} \mathrm{Tr} ((m^{\dag})^{\dag} n^{\dag}) 
  	= (w, v) \mathrm{Tr} (mn^{\dag}) = (\phi, \psi),
  \end{equation}
  therefore $J$ is an antiunitary operator. From $J (v \otimes m) = j (v)
  \otimes \mathfrak{j} (m) = j (v) \otimes m^{\dag}$, there is
    \begin{equation}
    J^2  (v \otimes m) = J (j (v) \otimes m^{\dag}) = j^2 (v) \otimes m = - v
    \otimes m,
  \end{equation}
  hence $J^2 = - 1$.
  
  For the grading operator $\Gamma$, since $\Gamma (v \otimes m) =
  \gamma_5 v \otimes m$ and $\gamma_5^2 = \mathbb{I}$, we have $\Gamma^2 = 1$. $\Gamma$
  acts only on the spinor part and commutes with the algebra action: $a \Gamma
  = \Gamma a$. Because $\gamma_5$ anticommutes with all $\gamma^\mu$, and $D$ is
  built from $\gamma^\mu$ and $\gamma_5$, one immediately obtain $\Gamma D + D
  \Gamma = 0$. The relation $j \gamma_5 = \gamma_5 j$ lets to $J \Gamma = \Gamma J$.

  Next we verify the reality condition $JD = DJ$.
  \begin{eqnarray}
&&    (d_\mu m)^{\dag} = [X_\mu, m]^{\dag} = - [X_\mu, m^{\dag}] = - d_\mu (m^{\dag}),
\nonumber\\
&&    (\delta_\mu m)^{\dag} = \{X_\mu, m\}^{\dag} = \{X_\mu, m^{\dag} \} = \delta_\mu
    (m^{\dag}),
  \end{eqnarray}
  hence
  \begin{eqnarray}
    J (\mathrm{i} \gamma^\mu \otimes d_\mu)  (v \otimes m) & = & J (\mathrm{i} \gamma^\mu v \otimes
    d_\mu m) \nonumber\\
    & = & j (\mathrm{i} \gamma^\mu v) \otimes (d_\mu m)^{\dag} \nonumber\\
    & = & \mathrm{i} \gamma^\mu j (v) \otimes d_\mu (m^{\dag}) \nonumber\\
    & = & (\mathrm{i} \gamma^\mu \otimes d_\mu)  (j (v) \otimes m^{\dag}) \nonumber\\
    & = & (\mathrm{i} \gamma^\mu \otimes d_\mu) J (v \otimes m). 
  \end{eqnarray}
 Similarly, there is
  \begin{equation}
    J (\gamma_5 \gamma^\mu \otimes \delta_\mu)  (v \otimes m) = (\gamma_5 \gamma^\mu \otimes \delta_\mu) J (v \otimes m) 
  \end{equation}
 It is easy to see that, there is $JD = DJ$.
  
  The verification of the order-zero condition is as follows. By definition,
  $Jb^{\dag} J^{- 1}  (v \otimes m) = v \otimes mb = \mathrm{id} \otimes R_b$,
  where $R_b$ is right multiplication by $b$. For any $a, b \in A$, there is
  \begin{equation}
    a(Jb^{\dag} J^{- 1}) (v \otimes m) = a (v \otimes mb)= v \otimes (a mb)=(Jb^{\dag} J^{- 1}) a (v \otimes m)
  \end{equation}
So $Jb^{\dag} J^{- 1}$ commutes with $a$, and $[a, Jb^{\dag} J^{- 1}] = 0$.
  
Finally, let us verify the order-one condition.
For the Dirac operator $D_1$, we have
\begin{equation}
	[D_1,a] = \mathrm{i} \sum_{\mu = 1}^4 \gamma^{\mu} \otimes [d_{\mu},a].
\end{equation}
Since there is
\begin{equation}
    (\gamma^{\mu} \otimes [d_{\mu},a])(Jb^{\dag} J^{- 1}) (v \otimes m) 
    = \gamma^{\mu} v \otimes ([d_{\mu},a]mb)
    =(Jb^{\dag} J^{- 1})(\gamma^{\mu} \otimes [d_{\mu},a]) (v \otimes m),
  \end{equation}
so the operator $\gamma^{\mu} \otimes [d_{\mu},a]$ commutes with $Jb^{\dag} J^{- 1}$. It is easy to see that $[D_1,a]$ commutes with $Jb^{\dag} J^{- 1}$. Similarly, the operator $[D_2,a]$ also commutes with $Jb^{\dag} J^{- 1}$. Therefore, we have $[[D, a], Jb^{\dag} J^{- 1}] = 0$ for all $a, b \in A$.
  
From the above results, one can find that $(\mathcal{A}, \mathcal{H}, D, J, \Gamma)$ is a finite real even spectral triple with KO-dimension $4$.
\end{proof}

\begin{corollary}
  $(A, H, D_1, J, \Gamma)$ and $(A, H, D_2, J, \Gamma)$ are finite real even
  spectral triples of KO-dimension $4$.
\end{corollary}

A state $\psi$ on the algebra $A$ is a linear functional $\psi : A \to
\mathbb{C}$, satisfying positivity: $\psi (a^{\ast} a) \ge 0$ for all $a \in
A$, and of norm $1$. For a normal bounded quantum state $\psi$, it can also be
represented by a density operator $\rho \in A$, where $\rho$ is positive
semidefinite, self-adjoint, and of trace $1$. The state $\psi$ acting on an
element $a \in A$ can be defined as
\begin{equation}
  \psi (a) = \mathrm{tr} (\rho a) .
\end{equation}
Consider the states $\psi_1, \psi_2$ correspond to density matrices
$\rho_1, \rho_2$, respectively.

\begin{definition}
  {\cite{Connes}} In a spectral triple $(A, H, D)$, the Connes spectral
  distance between the states $\psi_1$ and $\psi_2$ is defined as
  \begin{equation}
    d (\psi_1, \psi_2) \equiv d (\rho_1, \rho_2) = \sup_{e \in B} |
    \mathrm{tr} (\rho_1 e) - \mathrm{tr} (\rho_2 e) |,
  \end{equation}
  where the set
  \begin{equation}
    B := \{e \in A : \|[D, e]\|_{op} \leqslant 1\},
  \end{equation}
and $\|a\|_{op}$ denotes the operator norm of $a$:
  \begin{equation}
    \|a\|_{op} \equiv \sup_{\xi \in H,   \| \xi \| = 1} \|a \xi
    \|, \qquad \|a\|^2 \equiv \mathrm{tr} (a^{\dag} a) .
  \end{equation}
\end{definition}

The map $L (a) = \|[D, a]\|_{op}$ is a Lipschitz seminorm. The inequality $L
(e) = \|[D, e]\|_{op} $ $\leqslant 1$ is the so-called ball
condition.

Obviously, there is $d (\psi_1, \psi_2) = d (\psi_2, \psi_1)$. If $e \in B$,
then $- e \in B$, hence we also have
\begin{equation}
  d (\rho_1, \rho_2) = \sup_{e \in B} [\mathrm{tr} (\rho_1 e) - \mathrm{tr}
  (\rho_2 e)] = \sup_{e \in B} \left[ \mathrm{tr} (\Delta \rho  
  e) \right] = \mathrm{tr} (\Delta \rho   e_o),
\end{equation}
where $e_o$ is the so-called optimal element. Since the
supremum in the Connes spectral distance can always be attained by an
Hermitian element \cite{Iochum}, one only need to consider an Hermitian optimal
element $e_o$.

The 1-forms are defined as the following $A$-bimodule:
\begin{equation}
  \Omega_D^1 (A) := \left\{ \left. \sum_i a_i [D, b_i] \right| a_i, b_i \in A
  \right\} .
\end{equation}
One can define the inner fluctuations of the Dirac operator
as \cite{Magee}
\begin{equation}
  D_{\omega} = D + \omega + \varepsilon' J \omega J^{- 1},
\end{equation}
where $\omega^{\ast} = \omega \in \Omega_D^1 (A)$ is a self-adjoint 1-form. If
the spectral triple also possesses a $\nu$-twisted real structure, then its
fluctuated Dirac operator is defined as
\begin{equation}
  D_{\omega} = D + \omega + \varepsilon' \nu J \omega J^{- 1} \nu = D + \omega
  + \varepsilon' \omega', \qquad \omega' = \nu J \omega J^{- 1} \nu .
\end{equation}
As shown in Ref.~\cite{Magee}, the real spectral triple $(A, H, D_{\omega}, J)$ is also
a $\nu$-twisted real spectral triple.

One can also define a conformal scaling transformation\cite{Dabrowski1}.
Take $k \in A$ to be positive and invertible, with bounded inverse
$k^{- 1}$. Denote $k_J := JkJ^{- 1}$,
\begin{equation}
  D_{k_J} = k_J Dk_J, \qquad \nu (h) = (k^{- 1} k_J) (h),
\end{equation}
and given a real
spectral triple $(A, H, D, J)$, then $(A, H, J, D_{k_J}, \nu)$ forms a $\nu$-twisted real spectral triple. If
$(A, H, D, J)$ is an even spectral triple with chirality operator $\Gamma$,
then $(A, H, D_{k_J}, J, \nu)$ is also an even spectral triple with chirality
operator $\Gamma$, and has the same KO-dimension as $(A, H, D, J)$.

\section{Connes Spectral Distance on the Fuzzy Torus}\label{sec3}
Now, let us study the Connes spectral distance of the states on a fuzzy torus. First, let us study some properties of the spectral distances and optimal elements.
\begin{lemma}
  For the operators $d_{\mu} (\cdot) \equiv [X_{\mu}, \cdot], \delta_{\mu}
  (\cdot) \equiv \{X_{\mu}, \cdot\} : A \rightarrow A$, $\mu = 1, 2, 3, 4$, we
  have
  \begin{equation}
    [d_{\mu}, a] = [\delta_{\mu}, a] = [X_{\mu}, a], \quad a \in A.
  \end{equation}
\end{lemma}

\begin{proof}
  Consider the action of $[d_{\mu}, a]$ on an arbitrary $m \in A$,
  \begin{eqnarray}
    {}[d_{\mu}, a] (m) & = & \left[ [X_{\mu},   \cdot
     ], a \right] (m) \nonumber\\
    & = & \left( [X_{\mu},   \cdot  ] a - a
    [X_{\mu},   \cdot  ]) m \right. \nonumber\\
    & = & [X_{\mu},   am  ] - a [X_{\mu},
      m  ] \nonumber\\
    & = & [X_{\mu}, a] (m), 
  \end{eqnarray}
  hence $[d_{\mu}, a] = [X_{\mu}, a]$. Similarly, there is $[\delta_{\mu}, a]
  = [X_{\mu}, a]$.
  
  \ 
\end{proof}

\begin{lemma}
  For the Dirac operator
  \begin{equation}
    D =  \mathrm{i} \alpha \sum_{\mu = 1}^4 \gamma^{\mu}
    \otimes d_{\mu} + \beta \sum_{\mu = 1}^4 \gamma_5 \gamma^{\mu} \otimes
    \delta_{\mu},
  \end{equation}
  there is
  \begin{equation}
    [D, a] = \sum_{\mu = 1}^4 (\mathrm{i} \alpha \gamma^{\mu} + \beta \gamma_5
    \gamma^{\mu}) \otimes [X_{\mu}, a], \quad a \in A.
  \end{equation}
\end{lemma}

\begin{proof}
  For any element $v \otimes m \in H$, both $d_{\mu}$ and $a \in A$ act only
  on $m$. Using the above lemma, we have
\begin{eqnarray}
	{}[D, a] & = & \left[ \mathrm{i} \alpha \sum_{\mu = 1}^4 \gamma^{\mu} \otimes d_{\mu}
	+ \beta \sum_{\mu = 1}^4 \gamma_5 \gamma^{\mu} \otimes \delta_{\mu}, a
	\right] \nonumber\\
	& = & \mathrm{i} \alpha \sum_{\mu = 1}^4 \gamma^{\mu} \otimes [d_{\mu}, a] + \beta
	\sum_{\mu = 1}^4 \gamma_5 \gamma^{\mu} \otimes [\delta_{\mu}, a] \nonumber\\
	& = & \mathrm{i} \alpha \sum_{\mu = 1}^4 \gamma^{\mu} \otimes [X_{\mu}, a] + \beta
	\sum_{\mu = 1}^4 \gamma_5 \gamma^{\mu} \otimes [X_{\mu}, a] \nonumber\\
	& = & \sum_{\mu = 1}^4 (\mathrm{i} \alpha \gamma^{\mu} + \beta \gamma_5
	\gamma^{\mu}) \otimes [X_{\mu}, a] . 
\end{eqnarray}
\end{proof}

\begin{lemma}
  For any $a \in A$, we have
  \begin{equation}\label{da}
    \|[D, a]\|_{op} = \lambda   \|D_a \|_{op},
  \end{equation}
  where
  \begin{equation}
    D_a = \sum_{\mu = 1}^4 \gamma^{\mu} \otimes [X_{\mu}, a], \quad \lambda =
    \sqrt{\alpha^2 + \beta^2} .
  \end{equation}
\end{lemma}

\begin{proof}
  Let $F = \mathrm{i} \alpha \mathbb{I}_4 + \beta \gamma_5$, then $F^{\dag} F =
  (\alpha^2 + \beta^2) \mathbb{I}_4$, and $\dfrac{F}{\sqrt{\alpha^2 + \beta^2}}
  = \dfrac{F}{\lambda}$ is a unitary matrix.
  \begin{eqnarray}
    {}[D, a] & = & \sum_{\mu = 1}^4 (\mathrm{i} \alpha \gamma^{\mu} + \beta \gamma_5
    \gamma^{\mu}) \otimes [X_{\mu}, a] \nonumber\\
    & = & (F \otimes \mathbb{I}) \left( \sum_{\mu = 1}^4 \gamma^{\mu} \otimes
    [X_{\mu}, a] \right) \nonumber\\
    & = & (F \otimes \mathbb{I}) D_a . 
  \end{eqnarray}
  Hence
  \begin{equation}
    \|[D, a]\|_{op} = \left\| \left( \dfrac{F}{\lambda} \otimes \mathbb{I} \right)
    \cdot \lambda D_a \right\|_{op} = \| \lambda D_a \|_{op} = \lambda \|D_a
    \|_{op} .
  \end{equation}
  
\end{proof}

The result (\ref{da}) means that the Dirac operators $D$ and $D_1, D_2$ have similar Lipschitz seminorm and ball conditions.

\begin{theorem}
  The spectral triples $(A, H, D)$ and $(A, H, \lambda D_1)$, $(A, H, \lambda
  D_2)$ have the same metric.
\end{theorem}

\begin{proof}
  Since $D = \alpha D_1 + \beta D_2$ with any real number $\alpha, \beta$.
  When $\alpha = 1, \beta = 0$, we have
  \begin{equation}
    \|[D_1, a]\|_{op} = \|D_a \|_{op}.
  \end{equation}
  When $\alpha = 0, \beta = 1$, we have
  \begin{equation}
    \|[D_2, a]\|_{op} = \|D_a \|_{op} .
  \end{equation}
  So there is
  \begin{eqnarray}
    \|[D, a]\|_{op} & = & \lambda \|D_a \|_{op} = \lambda \|[D_1, a]\|_{op} =
    \lambda \|[D_2, a]\|_{op}\nonumber\\
    & = & \|[\lambda D_1, a]\|_{op}\nonumber\\
    & = & \|[\lambda D_2, a]\|_{op} .
  \end{eqnarray}
  Therefore, the spectral triples $(A, H, D)$ and $(A, H, \lambda D_1)$, $(A,
  H, \lambda D_2)$ have the same ball condition, and hence the same metric.
  
  \ 
\end{proof}

 Denote by $d_D  (\varphi, \psi)$ the spectral distance between any two
  states $\varphi, \psi$ in the spectral triple $(A, H, D)$, then there is\cite{wly}
  \begin{equation}\label{td}
    d_{tD}  (\varphi, \psi) = \frac{1}{|t|} d_D  (\varphi, \psi),
  \end{equation}
  where $t$ is any non-zero real number.

\begin{corollary}
  In the spectral triples $(A, H, D)$ and $(A, H, D_1)$, $(A, H, D_2)$, the
  spectral distance between any two states $\varphi, \psi \in A$ satisfies
  \begin{equation}\label{dd12}
    d_D  (\varphi, \psi) = \frac{1}{\lambda} d_{D_1}  (\varphi, \psi) =
    \frac{1}{\lambda} d_{D_2}  (\varphi, \psi) .
  \end{equation}
\end{corollary}

\begin{corollary}
  In the spectral triples $(A, H, D)$ and $(A, H, D_1)$, $(A, H, D_2)$, the
  spectral distance between any two states $\varphi, \psi \in A$ satisfies
  \begin{equation}
    \left( \frac{1}{d_D (\varphi, \psi)} \right)^2 = \left(
    \frac{\alpha}{d_{D_1} (\varphi, \psi)} \right)^2 + \left(
    \frac{\beta}{d_{D_2} (\varphi, \psi)} \right)^2 .
  \end{equation}
\end{corollary}

Consider the Dirac operator
\begin{equation}
  \mathcal{D}= \mathrm{i} \sum_{\mu = 1}^4 \gamma^{\mu} \otimes d_{\mu} + \sum_{\mu =
  1}^4 \gamma_5 \gamma^{\mu} \otimes \delta_{\mu} = D_1 + D_2 .
\end{equation}
Then in the spectral triples $(A, H, \mathcal{D})$ and $(A, H, D_1)$, $(A, H,
D_2)$, the spectral distances satisfy the reciprocal Pythagorean theorem
\begin{equation}
  d_{\mathcal{D}}^{- 2}  (\varphi, \psi) = d_{D_1}^{- 2}  (\varphi, \psi) +
  d_{D_2}^{- 2}  (\varphi, \psi) .
\end{equation}
Thus, in some sense, $D_1$ built from commutators and $D_2$ built from
anticommutators can be regarded as two mutually orthogonal components of $\mathcal{D}$.

Since the spectral triples $(A, H, D)$ and $(A, H, D_1)$, $(A, H, D_2)$ have
similar metrics, in the following contents, we will mainly consider the
spectral distance in the spectral triple $(A, H, D_1)$.

\begin{lemma}
  In the spectral triple $(A, H, D_1)$, for any Hermitian element $a \in A$,
  there are
  \begin{eqnarray}
    &&\|[C, a]\|_{op} = \|[C^{\dag}, a]\|_{op} \leqslant \|[D_1, a]\|_{op},
\nonumber\\
    &&\|[S, a]\|_{op} = \|[S^{\dag}, a]\|_{op} \leqslant \|[D_1, a]\|_{op} .
  \end{eqnarray}
\end{lemma}

\begin{proof}
  For any Hermitian element $a \in A$, there are
  \begin{eqnarray}\label{mm1}
    D_a & = & \sum_{\mu = 1}^4 \gamma^{\mu} \otimes [X_{\mu}, a] \nonumber\\
    & = & \mathrm{i} \left( \begin{array}{cccc}
      0 & 0 & [X_1 + \mathrm{i}X_2, a] & [X_3 + \mathrm{i}X_4, a]\\
      0 & 0 & [- X_3 + \mathrm{i}X_4, a] & [X_1 - \mathrm{i}X_2, a]\\
      {}[X_1 - \mathrm{i}X_2, a] & [- X_3 - \mathrm{i}X_4, a] & 0 & 0\\
      {}[X_3 - \mathrm{i}X_4, a] & [X_1 + \mathrm{i}X_2, a] & 0 & 0
    \end{array} \right) \nonumber\\
    & = & \mathrm{i} \left( \begin{array}{cc}
      0 & M\\
      - M^{\dag} & 0
    \end{array} \right), 
  \end{eqnarray}
  where
  \begin{eqnarray}\label{mm}
    M & = & \left( \begin{array}{cc}
      {}[X_1 + \mathrm{i}X_2, a] & [X_3 + \mathrm{i}X_4, a]\\
      {}[- X_3 + \mathrm{i}X_4, a] & [X_1 - \mathrm{i}X_2, a]
    \end{array} \right) \nonumber\\
    & = & \left( \begin{array}{cc}
      {}[C, a] & [S, a]\\
      {}[- S^{\dag}, a] & [C^{\dag}, a]
    \end{array} \right) \nonumber\\
    & = & \left( \begin{array}{cc}
      {}[C, a] & [S, a]\\
      {}[S, a]^{\dag} & - [C, a]^{\dag}
    \end{array} \right) . 
  \end{eqnarray}
  Since $\|M^{\dag} M\|_{op} = \|MM^{\dag} \|_{op}$, there is
  \begin{equation}\label{dam}
    \|[D_1, a]\|_{op}^2 =   \|\mathrm{i}D_a \|_{op}^2 =  
    \|D_a^{\dag} D_a \|_{op} = \|M^{\dag} M\|_{op} .
  \end{equation}
  From the above matrix representation (\ref{mm}), we have
  \begin{equation}
    M^{\dag} M = \left( \begin{array}{cc}
      {}[C, a]^{\dag} [C, a] + [S, a] [S, a]^{\dag} & [C, a]^{\dag} [S, a] -
      [S, a] [C, a]^{\dag}\\
      {}[S, a]^{\dag} [C, a] - [C, a] [S, a]^{\dag} & [S, a]^{\dag} [S, a] +
      [C, a] [C, a]^{\dag}
    \end{array} \right) .
  \end{equation}
  By Bessel's inequality, we obtain
  \begin{equation}
    \sup_{\phi \in \mathfrak{h}, \langle \phi | \phi \rangle = 1}
    \!\!\!\!\! \langle \phi | [C, a]^{\dag} [C, a] +
    [S, a] [S, a]^{\dag} | \phi \rangle \leqslant \|M^{\dag} M\|_{op} =
     \|[D_1, a]\|_{op}^2 .
  \end{equation}
  Since $\langle \phi | [C, a]^{\dag} [C, a] | \phi \rangle \geqslant 0$ and
  $\langle \phi | [S, a] [S, a]^{\dag} | \phi \rangle \geqslant 0$, one can
  obtain
  \begin{eqnarray}
   && \sup_{\phi \in \mathfrak{h}, \langle \phi | \phi \rangle = 1}
    \!\!\!\!\! \langle \phi | [C, a]^{\dag} [C, a] |
    \phi \rangle \leqslant \|[D_1, a]\|_{op}^2,
\nonumber\\
   && \sup_{\phi \in \mathfrak{h}, \langle \phi | \phi \rangle = 1}
    \!\!\!\!\! \langle \phi | [S, a] [S, a]^{\dag} |
    \phi \rangle \leqslant \|[D_1, a]\|_{op}^2,
  \end{eqnarray}
  i.e.
  \begin{equation}
    \|[C, a]\|_{op}^2 \leqslant \|[D_1, a]\|_{op}^2, \qquad \|[S,
    a]\|_{op}^2 \leqslant \|[D_1, a]\|_{op}^2,
  \end{equation}
  or
  \begin{equation}
    \|[C, a]\|_{op} \leqslant \|[D_1, a]\|_{op}, \qquad
    \|[S, a]\|_{op} \leqslant \|[D_1, a]\|_{op} .
  \end{equation}
\end{proof}

In the above proof, we have used the Bessel's inequality \cite{Revisiting}: If the matrix
$M$ has entries $ M_{ij}$, then the inequality holds
\begin{equation}
  | M_{ij} |^2 \leqslant \sum_i | M_{ij} |^2 \leqslant \|M\|_{op}^2 .
\end{equation}

In general, it is difficult to derive the explicit expressions of the spectral distances for any states.
In the following content, we will first consider the spectral distances between some special types of states.

Define basis state vectors: $|i \rangle = (0, \ldots, 1, \ldots, 0)^T$ $(1 \leqslant i
\leqslant N)$, here the $i$-th element is 1, and the others are all 0. Then we have
\begin{eqnarray}
  && S| 1 \rangle = |2 \rangle, \quad S| 2 \rangle = |3 \rangle, \quad \ldots,\quad S |N-1 \rangle = |N \rangle,
  \quad S|N \rangle = |1 \rangle,
\nonumber\\
  && S^{\dag} |N \rangle = |N - 1 \rangle, \quad \ldots, \quad S^{\dag} |2
  \rangle = |1 \rangle, \quad S^{\dag} |1 \rangle = |N \rangle .
\end{eqnarray}
Any state on the fuzzy torus $(A, H, D_1)$ can be represented by a density
matrix $\rho \in A$ of the form
\begin{equation}
  \rho = \sum_{i, j} p_{ij} |i \rangle \langle j|,\quad p_{ij}\in\mathbb{C}, \quad \rho^{\dag} = \rho,
  \quad \mathrm{tr} (\rho) = 1.
\end{equation}
First, let us consider the spectral distance between any two diagonal states
\begin{equation}
  \psi = \sum_{i = 1}^N p_i |i \rangle \langle i|, \qquad \varphi = \sum_{i =
  1}^N q_i |i \rangle \langle i| .
\end{equation}
The spectral distance between $\psi$ and $\varphi$ is
\begin{eqnarray}\label{pqe}
  d (\psi, \varphi) & = & \sup_{e \in B} | \mathrm{tr} (\psi e) - \mathrm{tr}
  (\varphi e) | = \sup_{e \in B} | \mathrm{tr} [(\psi - \varphi) e] |
  \nonumber\\
  & = & \sup_{e \in B} \left| \mathrm{tr} \left( \sum_{i = 1}^N  (p_i - q_i)
  |i \rangle \langle i|e  \right) \right| \nonumber\\
  & = & \sup_{e \in B} \left| \sum_{i = 1}^N (p_i - q_i) e_{ii} \right| . 
\end{eqnarray}
One can see that the above result depends only on the diagonal entries of the
optimal element $e_o$.

In general, it is relatively difficult to obtain the optimal element directly.
But if the states and the optimal element are all diagonal, the calculations
will become much simpler.

For any element $a \in A = \mathbb{M}_N (\mathbb{C})$, define the conditional
expectation
\begin{equation}\label{ea}
  E (a) = \frac{1}{N}  \sum_{j = 1}^N C^j a C^{- j},
\end{equation}
Using the relation $1 + q + \cdots + q^{N - 1} = 0$, after some
straightforward calculations, one can obtain
\begin{equation}
  E (a) = \mathrm{diag} (a_{11}, a_{22}, \ldots, a_{NN}) .
\end{equation}
Thus $E (a)$ is just the diagonal part of the matrix $a$. Therefore, for
diagonal states $\psi, \varphi$, the spectral distance can also be written as
\begin{equation}
  d (\psi, \varphi) = \sup_{e \in B} \mathrm{tr} [(\psi - \varphi) E (e)] .
\end{equation}
\begin{lemma}\label{dea}
  In the spectral triple $(A, H, D_1)$, for any Hermitian element $a \in A =
   \mathbb{M}_N (\mathbb{C})$, we have
  \begin{equation}
    \|[D_1, E (a)]\|_{op} \leqslant \|[D_1, a]\|_{op} .
  \end{equation}
\end{lemma}

\begin{proof}
  Using the commutation relation $CS = qSC$, one can obtain $S = q^j C^{- j}
  SC^j$, $S^{\dag} = q^{- j} C^{- j} S^{\dag} C^j$. So there are
  \begin{eqnarray}
    [C, C^j a C^{- j}] = C^j  [C, a] C^{- j}, && [C^{\dag}, C^j a C^{- j}]
    = C^j  [C^{\dag}, a] C^{- j},
\nonumber\\{}
    [S, C^j a C^{- j}] = q^{- j} C^j  [S, a] C^{- j}, && [S^{\dag}, C^j a
    C^{- j}] = q^j C^j  [S^{\dag}, a] C^{- j}.
  \end{eqnarray}
From the expression (\ref{mm1}) and (\ref{mm}), we have
  \begin{equation}
    [D_1, C^j a C^{- j}] = \left( \begin{array}{cc}
      0 & - M'\\
      {M'}^{\dag} & 0
    \end{array} \right),
  \end{equation}
where
  \begin{eqnarray}
    M' & = & \left( \begin{array}{cc}
      [C, C^j aC^{- j}] & [S, C^j aC^{- j}]\\{}
      [- S^{\dag}, C^j aC^{- j}] & [C^{\dag}, C^j aC^{- j}]
    \end{array} \right) \nonumber\\
    & = & \left( \begin{array}{cc}
      C^j [C, a] C^{- j} & q^{- j} C^j [S, a] C^{- j}\\
      q^j C^j [- S^{\dag}, a] C^{- j} & C^j [C^{\dag}, a] C^{- j}
    \end{array} \right) \nonumber\\
    & = & U_j \left( \begin{array}{cc}
      [C, a] & [S, a]\\
      {}[- S^{\dag}, a] & [C^{\dag}, a]
    \end{array} \right) U_j^{\dag} \nonumber\\
    & = & U_j MU_j^{\dag}, 
  \end{eqnarray}
  where $U_j$ is the following unitary matrix
  \begin{equation}
    U_j = \left[ \left( \begin{array}{cc}
      1 & 0\\
      0 & q^j
    \end{array} \right) \otimes \mathbb{I}_N \right]  (\mathbb{I}_2 \otimes
    C^j) .
  \end{equation}
  Hence
  \begin{eqnarray}
    {}[D_1, C^j a C^{- j}] & = &  \left( \begin{array}{cc}
      0 & - U_j MU_j^{\dag}\\
      U_j M^{\dag} U_j^{\dag} & 0
    \end{array} \right) \nonumber\\
    & = & (\mathbb{I}_2 \otimes U_j)   \left(
    \begin{array}{cc}
      0 & - M\\
      M^{\dag} & 0
    \end{array} \right)  (\mathbb{I}_2 \otimes U_j^{\dag}) \nonumber\\
    & = & (\mathbb{I}_2 \otimes U_j) [D_1, a]  (\mathbb{I}_2 \otimes
    U_j^{\dag}) . 
  \end{eqnarray}
  Thus
  \begin{equation}
    \|[D_1, C^j a C^{- j}]\|_{op} = \|(\mathbb{I}_2
    \otimes U_j) [D_1, a] (\mathbb{I}_2 \otimes U_j^{\dag})\|_{op} = \|[D_1,
    a]\|_{op} .
  \end{equation}
  Using the triangle inequality for the operator norm, one can obtain
  \begin{eqnarray}
    \|[D_1, E (a)]\|_{op} & = & \left\| \left[ D_1, \frac{1}{N}  \sum_{j =
    1}^N C^j a C^{- j} \right] \right\|_{op} \nonumber\\
    & = & \left\| \frac{1}{N}  \sum_{j = 1}^N [D_1, C^j a C^{- j}]
    \right\|_{op} \nonumber\\
    & \leqslant & \frac{1}{N}  \sum_{j = 1}^N \|[D_1, C^j a C^{- j}]\|_{op}
    \nonumber\\
    & = & \frac{1}{N}  \sum_{j = 1}^N \|[D_1, a]\|_{op} \nonumber\\
    & = & \|[D_1, a]\|_{op} . 
  \end{eqnarray}
  
\end{proof}

The above lemma means that the condition expectation function (\ref{ea}) can lead to a contraction of the corresponding Lipschitz seminorm for the Dirac operator $D_1$, and also those for $D_2$ and $D$.

\begin{theorem}\label{th2}
  In the spectral triple $(A, H, D_1)$, the spectral distance between any two
  diagonal states $\psi, \varphi$ is
  \begin{equation}
    d (\psi, \varphi) = \sup_{e_{\mathrm{diag}} \in B}  \mathrm{tr} [(\psi -
    \varphi) e_{\mathrm{diag}}] = \mathrm{tr} [(\psi - \varphi) e_o] ,
  \end{equation}
  where $e_{\mathrm{diag}}$ denotes a real diagonal matrix, and the optimal
  element $e_o$ can be chosen to be a real diagonal matrix.
\end{theorem}

\begin{proof}
  Define the sets
  \begin{eqnarray}
    &&B := \{e \in A : e^{\dag} = e, \|[D_1, e]\|_{op} \leqslant 1\}, \nonumber\\
    &&B':= \{e \in A : e^{\dag} = e, \|[D_1, E (e)]\|_{op} \leqslant 1\} .
  \end{eqnarray}
  For any Hermitian element $e \in A$, $E (e) \in A$ is a real diagonal
  matrix. According to Lemma (\ref{dea}),
  \begin{equation}
    \|[D_1, E (e)]\|_{op} \leqslant \|[D_1, e]\|_{op} \leqslant 1,
  \end{equation}
  so for any $e \in B$, we have $e \in B'$, hence $B \subseteq B'$.
  Conversely, if $e \in B'$, then $E (e) \in B$. The spectral distance between
  diagonal states $\psi, \varphi$ is
  \begin{eqnarray}
    d (\psi, \varphi) & = & \sup_{e \in B} \mathrm{tr} [(\psi - \varphi) E
    (e)] \nonumber\\
    & \leqslant & \sup_{e \in B'} \mathrm{tr} [(\psi - \varphi) E (e)]
    \nonumber\\
    & = & \sup_{E (e) \in B} \mathrm{tr} [(\psi - \varphi) E (e)]
    \nonumber\\
    & = & \sup_{e_{\mathrm{diag}} \in B} \mathrm{tr} [(\psi - \varphi)
    e_{\mathrm{diag}}] \nonumber\\
    & \leqslant & \sup_{e \in B} \mathrm{tr} [(\psi - \varphi) e] = d
    (\psi, \varphi), 
  \end{eqnarray}
  where $e_{\mathrm{diag}}$ denotes a real diagonal matrix. So we have
  \begin{equation}
    d (\psi, \varphi) = \sup_{e_{\mathrm{diag}} \in B} \mathrm{tr} [(\psi -
    \varphi) e_{\mathrm{diag}}] .
  \end{equation}
  This means that, to compute the spectral distance between diagonal states $\psi,
  \varphi$, one only needs to consider real diagonal matrices $e$ satisfying
  the ball condition $\|[D_1, e]\|_{op} \leqslant 1$.
  
  \ 
\end{proof}

\begin{lemma}
  For any real diagonal matrix $e = \mathrm{diag} (e_1, \ldots, e_N) \in A$,
  the Lipschitz seminorm $\|[S, e]\|_{op}$ is
  \begin{equation}\label{se}
    \|[S, e]\|_{op} = \max_{k = 1}^N |e_{k + 1} - e_k |,
  \end{equation}
  where $e_{N + 1} \equiv e_1$.
\end{lemma}

\begin{proof}
Using the matrix representation (\ref{cs}), after some straightforward calculations, one can obtain
  \begin{equation}
    S^{\dagger} eS = \mathrm{diag} (e_2, e_3, \ldots, e_N, e_1).
  \end{equation}
So there is
  \begin{eqnarray}
    {}[S, e]^{\dag} [S, e] & = & - [S^{\dag}, e] SS^{\dag} [S, e] \nonumber\\
    & = & - (S^{\dag} eS - e)  (e - S^{\dag} eS) \nonumber\\
    & = & (S^{\dag} eS - e)^2 \nonumber\\
    & = & \mathrm{diag} [(e_2 - e_1)^2, (e_3 - e_2)^2, \ldots, (e_N - e_{N -
    1})^2, (e_1 - e_N)^2] . 
  \end{eqnarray}
  Hence
  \begin{eqnarray}
    \|[S, e]\|_{op} & = & \max \{|e_2 - e_1 |, |e_3 - e_2 |, \ldots, |e_N -
    e_{N - 1} |, |e_1 - e_N |\} \nonumber\\
    & = & \max_{k = 1}^N |e_{k + 1} - e_k |, 
  \end{eqnarray}
  where $e_{N + 1} \equiv e_1$.
\end{proof}

\begin{corollary}
In the spectral triple $(A, H, D_1)$, for any real diagonal matrix $e =
  \mathrm{diag} (e_1,$ $\ldots, e_N) \in A$, the Lipschitz seminorm is
  \begin{equation}
    \|[D_1, e]\|_{op} = \max_{k = 1}^N |e_{k + 1} - e_k |,
  \end{equation}
  and the ball condition inequality can be written as
  \begin{equation}
    \max_{k = 1}^N |e_{k + 1} - e_k | \leqslant 1.
  \end{equation}
\end{corollary}

\begin{proof}
  For any diagonal element $e$, from (\ref{cs}) and (\ref{mm}), there are
  \begin{equation}
    [C, e] = 0, \qquad M = \left( \begin{array}{cc}
      0 & [S, e]\\
      {}[S, e]^{\dag} & 0
    \end{array} \right) .
  \end{equation}
  Using the above results (\ref{dam}) and (\ref{se}), one can obtain
  \begin{eqnarray}
    \|[D_1, e]\|_{op} & = & \|M\|_{op} \nonumber\\
    & = & \|[S, e]\|_{op} \nonumber\\
    & = & \max_k |e_{k + 1} - e_k |. 
  \end{eqnarray}
Furthermore, using the ball condition $\|[D_1, e]\|_{op}\leqslant 1$, there is
\begin{equation}
	\max_{k = 1}^N |e_{k + 1} - e_k | \leqslant 1.
\end{equation}
\end{proof}

\begin{theorem}\label{th3}
  In the spectral triple $(A, H, D_1)$, consider any two diagonal states
  \begin{equation}\label{pf}
    \psi = \sum_{i = 1}^N p_i |i \rangle \langle i|, \qquad \varphi = \sum_{i
    = 1}^N q_i |i \rangle \langle i|,
  \end{equation}
  then the spectral distance between $\psi$ and $\varphi$ is
  \begin{equation}\label{sd}
    d (\psi, \varphi) = \sup_{e \in B_o} \mathrm{tr} [(\psi - \varphi) e]
    = \sup_{e \in B_o} \sum_{i = 1}^N (p_i - q_i) e_i,
  \end{equation}
  where the set
  \begin{equation}
    B_o = \left\{ e = \mathrm{diag} (e_1, \ldots, e_N) \in A, e_i \in
    \mathbb{R}: \max_{k = 1}^N |e_{k + 1} - e_k | = 1 \right\} .
  \end{equation}
\end{theorem}

\begin{proof}
  According to the result (\ref{pqe}) and Theorem (\ref{th2}), the spectral distance between the
  diagonal states $\psi, \varphi$ is
  \begin{eqnarray}
    d (\psi, \varphi) & = & \sup_{e \in B_d} \mathrm{tr} [(\psi - \varphi)
    e] \nonumber\\
    & = & \sup_{e \in B_d} \sum_{i = 1}^N (p_i - q_i) e_i 
    \nonumber\\
    & = & \mathrm{tr} [(\psi - \varphi) e_o]  < + \infty, 
  \end{eqnarray}
  where
  \begin{equation}
    B_d = \left\{ e = \mathrm{diag} (e_1, \ldots, e_N) \in A, e_i \in
    \mathbb{R}: \max_{k = 1}^N |e_{k + 1} - e_k | \leqslant 1 \right\} .
  \end{equation}
  The optimal element $e_o$ is a real diagonal matrix, and its entries satisfy
  $\max_{k} |e_{k + 1} - e_k | \leqslant 1$. If the diagonal entries of the optimal element $e_o$ satisfy
  \begin{equation}
    \max_{k = 1}^N |e_{k + 1} - e_k | = t < 1,
  \end{equation}
  then one can take $e_o' =\mathrm{diag} (e_1', \ldots, e_N')= \frac{1}{t} e_o$. Therefore, there is
  \begin{equation}
    \max_{k = 1}^N |e_{k + 1}' - e_k' | = 1,
  \end{equation}
  and we still have $e_o' \in B_d$. However,
  \begin{eqnarray}
    \mathrm{tr} [(\psi - \varphi) e_o'] & = & \frac{1}{t} 
    \mathrm{tr} [(\psi - \varphi) e_o] \nonumber\\
    & > & \mathrm{tr} [(\psi - \varphi) e_o] \nonumber\\
    & = & \sup_{e \in B_d} \mathrm{tr} [(\psi - \varphi) e] . 
  \end{eqnarray}
  This is a contradiction. Therefore, the diagonal entries of the optimal
  element $e_o$ must satisfy
  \begin{equation}
    \max_{k = 1}^N |e_{k + 1} - e_k | = 1.
  \end{equation}
\end{proof}

\section{Some simple examples}\label{sec4}
From Theorem (\ref{th3}), one can calculate the spectral distances
of any diagonal states, and only need to consider the optimal elements being real diagonal matrices.
But it is usually still very cumbersome to derive the explicit results.
In the following content, we will only explicitly calculate the spectral distances
of some simple examples.

First, let us consider the basic states $|i \rangle$, and the correspond density matrices are:
\begin{equation}
  \omega_i = |i \rangle \langle i|, \qquad 1\leqslant i \leqslant N.
\end{equation}
In the spectral triple $(A, H, D_1)$, for any $1 \leqslant i < N$, consider
the spectral distance between the basis states $|i \rangle, |i + 1 \rangle$. According to the above formula (\ref{sd}), there is
\begin{eqnarray}
  d (|i \rangle, |i + 1 \rangle) & = & d (\omega_i, \omega_{i + 1})
  \nonumber\\
  & = & \sup_{e \in B_o} \sum_{i = 1}^N (p_i - q_i) e_i
  \nonumber\\
  & = & e_i - e_{i + 1} \nonumber\\
  & = & 1, 
\end{eqnarray}
and the optimal element $e_o$ can be taken as
\begin{equation}
  e_o = \frac{1}{2} (|i \rangle \langle i| - |i + 1 \rangle  \langle
  i + 1|) .
\end{equation}

From Theorem (\ref{th3}), it is easy to see that, there are some kinds of cyclic symmetry in both the optimal elements and the spectral distances between the diagonal states, such as $|1 \rangle, |2 \rangle, \ldots, |N \rangle$.

\begin{corollary}
  The spectral distances between the basis states of the fuzzy torus $(A, H,
  D_1)$ are
  \begin{equation}
    d (|1 \rangle, |2 \rangle) = d (|2 \rangle, |3 \rangle) = \ldots = d (|N
    \rangle, |1 \rangle) = 1.
  \end{equation}
\end{corollary}

For any $2 \leqslant i \leqslant \frac{N}{2} + 1$, the spectral distance
between the basis states $|1 \rangle$ and $|i \rangle$ of the fuzzy
torus $(A, H, D_1)$ is
\begin{eqnarray}
  d (|1 \rangle, |i \rangle) = d (\omega_1, \omega_i) & = & \sup_{e \in B_o}
  \sum_{i = 1}^N (p_i - q_i) e_i \nonumber\\
  & = & e_1 - e_i  \nonumber\\
  & \leqslant & \sum_{t = 1}^{i - 1} |e_t - e_{t + 1} | \nonumber\\
  & \leqslant & i - 1. 
\end{eqnarray}
So one can take the optimal element $e_o = \mathrm{diag} (e_1, \ldots, e_N)$
with entries as
\begin{eqnarray}
&&  e_1 =  i - 1, \quad e_2 = i - 2, \quad
  \ldots, \quad e_{i - 1} = 1 , \quad e_i = 0,
\nonumber\\
&&  e_{i + 1} = \frac{i - 1}{N + 1 - i}, \quad e_{i + 2} =
  \frac{2 (i - 1)}{N + 1 - i}, \quad \ldots, \quad e_N =
  \frac{(N - i)  (i - 1)}{N + 1 - i} .
\end{eqnarray}
Hence there is
\begin{equation}
  d (|1 \rangle, |i \rangle) = e_1 - e_i = i - 1.
\end{equation}
\begin{corollary}
  For any $2 \leqslant i \leqslant \frac{N}{2} + 1$, the spectral distances
  between the basis states of the fuzzy torus $(A, H, D_1)$ are
  \begin{equation}
    d (|1 \rangle, |i \rangle) = d (|2 \rangle, |i + 1 \rangle) = \ldots = d
    (|N - i + 1 \rangle, |N \rangle) = i - 1.
  \end{equation}
\end{corollary}

One can see that the spectral distances between basis states of
the fuzzy torus do not depend on $N$.

\begin{corollary}
  For any $1 \leqslant i \leqslant j \leqslant N$, the spectral distance between the basis states
  $|i \rangle$ and $|j \rangle$ of the fuzzy torus $(A, H, D_1)$ is
  \begin{equation}\label{ij}
    d (|i \rangle, |j \rangle) = d (|j \rangle, |i \rangle) = \min \left\{ |i - j|, N - |i - j| \right\}.
  \end{equation}
\end{corollary}

Next, let us consider the spectral distances between some simple mixed diagonal states.
Let us first consider the following states, for any $1 \leqslant i \leqslant j \leqslant N$,
\begin{equation}
	\psi =  p|i \rangle \langle i|+(1-p)|j\rangle \langle j|, \quad \varphi = q|i \rangle \langle i|+(1-q)|j\rangle \langle j|, \quad 0 \leqslant p, q \leqslant 1.
\end{equation}
The spectral distance between $\psi$ and $\varphi$ is
\begin{eqnarray}
	d (\psi, \varphi) & = & \sup_{e \in B_o} \left| \sum_{i = 1}^N (p_i - q_i)
	e_i \right| \nonumber\\
	& = & |p - q|\cdot \sup_{e \in B_o} |e_i - e_j| .
\end{eqnarray}
By virtue of the above result (\ref{ij}), it is easy to see that, there is
\begin{equation}
	d (\psi, \varphi) = | p - q | \cdot \min \left\{ |i - j|, N - |i - j| \right\}.
\end{equation}

\begin{proposition}\label{pr2}
  In the fuzzy torus $(A, H, D_1)$, consider two diagonal states
  \begin{equation}
    \psi = p|i \rangle \langle i|+(1-p)|i+1 \rangle \langle i+1|, \quad 0 \leqslant p\leqslant1,1\leqslant i< N, \qquad \varphi
    = \frac{1}{N} \mathbb{I}_N.
  \end{equation}
  When $N$ is even, the spectral distance between $\psi$ and $\varphi$ is
  \begin{equation}
    d (\psi, \varphi) = \left\{ \begin{array}{ll}
      \dfrac{N}{4} - p, & 0 \leqslant p
      \leqslant \dfrac{1}{2} ;\\[1em]
      \dfrac{N}{4} - 1 + p , & \dfrac{1}{2}
      \leqslant p \leqslant 1.
    \end{array} \right.
  \end{equation}
  When $N$ is odd, the spectral distance between $\psi$ and $\varphi$ is
  \begin{equation}
    d (\psi, \varphi) = \left\{ \begin{array}{ll}
      \dfrac{N^2 - 1}{4 N} - p , & 0
      \leqslant p \leqslant \dfrac{N - 1}{2 N} ;\\[1em]
      \dfrac{(N - 1)^2}{4 N}, & \dfrac{N - 1}{2 N} \leqslant
      p \leqslant \dfrac{N + 1}{2 N} ;\\[.5em]
      \dfrac{N^2 - 1}{4 N} - 1 + p , &
      \dfrac{N + 1}{2 N} \leqslant p \leqslant 1.
    \end{array} \right.
  \end{equation}
\end{proposition}
\begin{proof}
	See the appendix for a detailed proof. 
\end{proof}
For example, when $N = 3$,
there are
\begin{equation}
  \psi = \mathrm{diag} (p, 1 - p, 0), \quad \varphi = \frac{1}{3}
  \mathrm{diag} (1, 1, 1), \quad 0 \leqslant p \leqslant 1.
\end{equation}
The spectral distance between $\psi$ and $\varphi$ is
\begin{equation}
  d (\psi, \varphi) = \left\{ \begin{array}{ll}
    \dfrac{2}{3} - p , & 0 \leqslant p
    \leqslant \dfrac{1}{3} ;\\[1em]
     \dfrac{1}{3}, & \dfrac{1}{3} \leqslant p \leqslant
    \dfrac{2}{3} ;\\[1em]
    p - \dfrac{1}{3}, & \dfrac{2}{3}
    \leqslant p \leqslant 1.
  \end{array} \right.
\end{equation}
The spectral distances of some states in this case are depicted in Figure \ref{fig1}.

\begin{figure}
	\centering
	\includegraphics[width=0.5\textwidth]{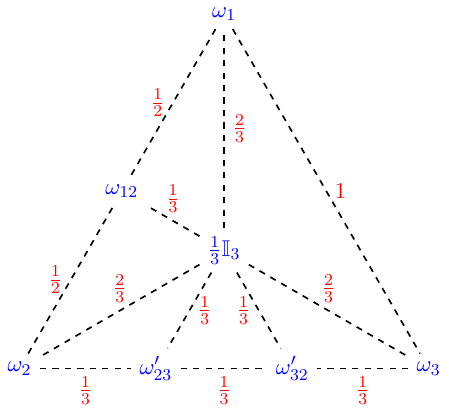}
	\caption{\label{fig1}Connes spectral distances between some states in the case $N=3$, where
		$\omega_i=|i\rangle \langle i|$, $\omega_{ij}=\frac{1}{2}(|i\rangle \langle i|+|j\rangle \langle j|)$, and $\omega_{ij}'=\frac{2}{3}|i\rangle \langle i|+\frac{1}{3}|j\rangle \langle j|$.}
\end{figure}

When $N = 4$, there are
\begin{equation}
  \psi = \mathrm{diag} (p, 1 - p, 0, 0), \quad \varphi = \frac{1}{4}
  \mathrm{diag} (1, 1, 1, 1), \quad 0 \leqslant p \leqslant 1.
\end{equation}
Then the spectral distance between $\psi$ and $\varphi$ is
\begin{equation}
  d (\psi, \varphi) = \left\{ \begin{array}{ll}
    1 - p, & 0 \leqslant p \leqslant \dfrac{1}{2}
    ;\\[1em]
     p, & \dfrac{1}{2} \leqslant p \leqslant 1.
  \end{array} \right.
\end{equation}

For the spectral distance between arbitrary non-diagonal states, it is
generally very difficult to calculate. For simplicity, we only consider the
case $N = 2$. In this case,
\begin{equation}
  C = \sigma_3 = \left( \begin{array}{cc}
    1 & 0\\
    0 & - 1
  \end{array} \right), \qquad S = \sigma_1 = \left( \begin{array}{cc}
    0 & 1\\
    1 & 0
  \end{array} \right),
\end{equation}
Obviously, any matrix $e \in A$ and $e +\mathbb{I}$ satisfy the same ball
condition and distance formula
\begin{equation}
  \|[D_1, e]\|_{op} = \|[D_1, e +\mathbb{I}]\|_{op} \leqslant 1, \quad
  \mathrm{tr} (\Delta \rho e) = \mathrm{tr} [\Delta \rho (e +\mathbb{I})],
\end{equation}
so one can only consider traceless optimal elements, $\mathrm{tr} (e_o) = 0$.

Any $2 \times 2$ traceless Hermitian matrix $e$ can be expressed by the Pauli
matrices as
\begin{equation}
  e = \vec{e} \cdot \vec{\sigma} = e_1 \sigma_1 + e_2 \sigma_2 + e_3 \sigma_3,
\end{equation}
where $\vec{e} = (e_1, e_2, e_3)$ is a real vector, $\vec{\sigma} = (\sigma_1,
\sigma_2, \sigma_3)$, and $\sigma_i$ are the Pauli matrices. Using the
following formulas of the Pauli matrices
\begin{eqnarray}
  &&\sigma_i \sigma_j = \delta_{ij} \mathbb{I}+ \mathrm{i} \varepsilon_{ijk}
  \sigma_k,
\nonumber\\
  &&(\vec{a} \cdot \vec{\sigma})  (\vec{b} \cdot \vec{\sigma}) = (\vec{a} \cdot
  \vec{b}) \mathbb{I}_2 + \mathrm{i} (\vec{a} \times \vec{b}) \cdot
  \vec{\sigma},
\nonumber\\
  &&[\vec{a} \cdot \vec{\sigma}, \vec{b} \cdot \vec{\sigma}] = 2 \mathrm{i}
  (\vec{a} \times \vec{b}) \cdot \vec{\sigma},
\end{eqnarray}
one can obtain
\begin{eqnarray}
  &&[C, e] = [\sigma_3, \vec{e} \cdot \vec{\sigma}] = 2 \mathrm{i} (e_1 \sigma_2 -
  e_2 \sigma_1),
\nonumber\\
  &&[S, e] = [\sigma_1, \vec{e} \cdot \vec{\sigma}] = 2 \mathrm{i} (e_2 \sigma_3 -
  e_3 \sigma_2),
\end{eqnarray}
and then
\begin{eqnarray}
  M & = & \left( \begin{array}{cc}
    {}[C, e] & [S, e]\\
    {}[S, e]^{\dag} & - [C, e]^{\dag}
  \end{array} \right) \nonumber\\
  & = & 2 \mathrm{i} [\mathbb{I}_2 \otimes (e_1 \sigma_2 - e_2 \sigma_1) + \mathrm{i}
  \sigma_2 \otimes (e_2 \sigma_3 - e_3 \sigma_2)] . 
\end{eqnarray}
After some straightforward calculations, one can derive
\begin{equation}
  M^{\dagger} M = 4 (e_1^2 + 2 e_2^2 + e_3^2) \mathbb{I}_4 - 8 e_2  (\sigma_2
  \otimes e) .
\end{equation}
The eigenvalues of $e = \vec{e} \cdot \vec{\sigma}$ are $\{| \vec{e} |, - |
\vec{e} |\}$, and then the eigenvalues of $\sigma_2 \otimes e$ are $\{|
\vec{e} |, | \vec{e} |, - | \vec{e} |, - | \vec{e} |\}$. So the eigenvalues of
$M^{\dagger} M$ are
\[ 4 (e_1^2 + 2 e_2^2 + e_3^2) \pm 8 e_2 | \vec{e} | = 4 (| \vec{e} | \pm
   e_2)^2 \]
Therefore, we have
\begin{equation}
  \|M\|_{op} = 2 (| \vec{e} | + | e_2 |) = 2 (|e_2 | + \sqrt{e_1^2 + e_2^2 +
  e_3^2}) = 2 (|e_2 | + \sqrt{e_2^2 + \theta^2}),
\end{equation}
where $\theta = \sqrt{e_1^2 + e_3^2}$. Using the ball condition, one can
obtain the inequality
\begin{equation}
  |e_2 | + \sqrt{e_2^2 + \theta^2} \leqslant \frac{1}{2} .
\end{equation}
Obviously, there are
\begin{equation}
  |e_2 | \leqslant \frac{1}{4}, \qquad \theta \leqslant \frac{1}{2} .
\end{equation}

In general, a density matrix $\rho \in  \mathbb{M}_2 (\mathbb{C})$ can be uniquely
represented by a Bloch vector {\cite{Nielsen}}
\begin{equation}\label{rr}
  \rho = \frac{I + \vec{r} \cdot \vec{\sigma}}{2},
\end{equation}
where the real vector $\vec{r} = (x, y, z)$ is called the Bloch vector, $|
\vec{r} | \leqslant 1$.

Consider two states $\rho_1$ and $\rho_2$ with corresponding Bloch vectors
$\vec{r}_1 = (x_1, y_1, z_1)$ and $\vec{r}_2 = (x_2, y_2, z_2)$. Using the above
matrix representation (\ref{rr}), one can obtain
\begin{equation}
  \Delta \rho = \rho_1 - \rho_2 = \frac{1}{2} \Delta \vec{r} \cdot
  \vec{\sigma},
\end{equation}
where $\Delta \vec{r} = \vec{r}_1 - \vec{r}_2 = (\Delta x, \Delta y, \Delta
z)$, and $\Delta x = x_1 - x_2$, $\Delta y = y_1 - y_2$, $\Delta z = z_1 -
z_2$. Therefore, we have
\begin{equation}
  \mathrm{tr} (\Delta \rho   e) = \mathrm{tr} \left[ \frac{1}{2}
  (\Delta \vec{r} \cdot \vec{\sigma}) (\vec{e} \cdot \vec{\sigma}) \right] =
  \Delta \vec{r} \cdot \vec{e} = e_1 \Delta x + e_2 \Delta y + e_3 \Delta z.
\end{equation}
The spectral distance between the states $\rho_1$ and $\rho_2$ is
\begin{eqnarray}\label{d12}
  d (\rho_1, \rho_2) & = & \sup_{e \in B} | \mathrm{tr} (\Delta \rho
    e) | \nonumber\\
  & = & \sup_{e \in B} |e_1 \Delta x + e_2 \Delta y + e_3 \Delta z|
  \nonumber\\
  & \leqslant & \sup_{e \in B} (|e_2 \Delta y| + |e_1 \Delta x + e_3 \Delta
  z|) \nonumber\\
  & \leqslant & \sup_{e \in B} \left( |e_2 | \cdot | \Delta y| + \sqrt{e_1^2
  + e_3^2}  \sqrt{(\Delta x)^2 + (\Delta z)^2} \right) \nonumber\\
  & = & \sup_{e \in B} \left( |e_2 | \cdot | \Delta y| + \theta \sqrt{(\Delta
  x)^2 + (\Delta z)^2} \right) . 
\end{eqnarray}
In the last inequality, we used the Cauchy-Schwarz inequality, $(a_1 b_1 + a_2
b_2)^2 \leqslant (a_1^2 + a_2^2)  (b_1^2 + b_2^2)$, where equality holds if
and only if $a_1 b_2 = a_2 b_1$. Therefore, one can choose the optimal element
$e$ to satisfy
\begin{equation}
  e_3 \Delta x = e_1 \Delta z.
\end{equation}
Obviously, for given states $\rho_1$ and $\rho_2$, in order to maximize $|
\mathrm{tr} (\Delta \rho   e) |$, one should choose an Hermitian
element $e$ such that $|e_2 |$ and $\theta = \sqrt{e_1^2 + e_3^2}$ are as
large as possible. Hence in this case we have
\begin{equation}
  |e_2 | + \sqrt{e_2^2 + \theta^2} = \frac{1}{2},
\end{equation}
or
\begin{equation}\label{e2}
  |e_2 | = \frac{1}{4} - \theta^2 .
\end{equation}
Using the above equations (\ref{d12}) and (\ref{e2}), there are
\begin{eqnarray}
  d (\rho_1, \rho_2) & \leqslant & \sup_{e \in B} \left( |e_2 | \cdot | \Delta
  y| + \theta \sqrt{(\Delta x)^2 + (\Delta z)^2} \right) \nonumber\\
  & = & \sup_{e \in B} \left( \frac{1}{4 } | \Delta y| - 
  \theta^2 | \Delta y| + \theta \sqrt{(\Delta x)^2 + (\Delta z)^2} \right)
  \nonumber\\
  & = & \sup_{e \in B} \left[ \frac{| \Delta \vec{r} |^2}{4 | \Delta
  y|} - | \Delta y| \left( \theta - \frac{\sqrt{(\Delta x)^2 +
  (\Delta z)^2}}{2 | \Delta y|} \right)^2 \right], 
\end{eqnarray}
where $| \Delta \vec{r} | = \sqrt{(\Delta x)^2 + (\Delta y)^2 + (\Delta z)^2}
\leqslant 2$.

The ball condition $e \in B$ implies $\theta \leqslant \frac{1}{2}$.
When $(\Delta x)^2 + (\Delta z)^2 \leqslant (\Delta y)^2$, one can take the
optimal element $e$ to satisfy
\begin{equation}
  \theta = \frac{\sqrt{(\Delta x)^2 + (\Delta z)^2}}{2 | \Delta y|},
  \qquad |e_2 | = \frac{(\Delta y)^2 - (\Delta x)^2 - (\Delta z)^2}{4
  (\Delta y)^2} .
\end{equation}
Taking into account $e_3 \Delta x = e_1 \Delta z$, one can for instance take
\begin{equation}
  e_1 = \frac{\Delta x}{2 | \Delta y|}, \quad e_2 = \frac{(\Delta
  y)^2 - (\Delta x)^2 - (\Delta z)^2}{4  (\Delta y)^2}, \quad e_3 =
  \frac{\Delta z}{2 | \Delta y|},
\end{equation}
and the optimal element
\begin{eqnarray}
  e_o & = & e_1 \sigma_1 + e_2 \sigma_2 + e_3 \sigma_3 \nonumber\\
  & = & \left( \begin{array}{cc}
    \frac{\Delta z}{2 | \Delta y|} & \frac{\Delta x}{2 |
    \Delta y|} - \mathrm{i} \frac{(\Delta y)^2 - (\Delta x)^2 - (\Delta
    z)^2}{4 (\Delta y)^2}\\
    \frac{\Delta x}{2 | \Delta y|} + \mathrm{i} \frac{(\Delta y)^2 -
    (\Delta x)^2 - (\Delta z)^2}{4 (\Delta y)^2} & - \frac{\Delta
    z}{2 | \Delta y|}
  \end{array} \right) \nonumber\\
  & = & \frac{1}{2 | \Delta y|} \left( \begin{array}{cc}
    \Delta z & \mathrm{i} \frac{(\Delta x - \mathrm{i}| \Delta y|)^2 + (\Delta z)^2}{2|
    \Delta y|}\\
    - \mathrm{i} \frac{(\Delta x + \mathrm{i}| \Delta y|)^2 + (\Delta z)^2}{2| \Delta
    y|} & - \Delta z
  \end{array} \right) . 
\end{eqnarray}
In this case we have
\begin{equation}
  d (\rho_1, \rho_2) = \frac{| \Delta \vec{r} |^2}{4 | \Delta y|} .
\end{equation}
When $(\Delta x)^2 + (\Delta z)^2 \geqslant (\Delta y)^2$, one can only take the
optimal element $e$ satisfying
\begin{equation}
  \theta = \frac{1}{2} \leqslant \frac{\sqrt{(\Delta x)^2 + (\Delta z)^2}}{2
  | \Delta y|}, \qquad |e_2 | = \frac{1}{4 } - \theta^2 = 0.
\end{equation}
Taking into account $w \Delta x = u \Delta z$, one can for instance take
\begin{equation}
  e_1 = \frac{1}{2}  \frac{\Delta x}{\sqrt{(\Delta x)^2 + (\Delta
  z)^2}}, \quad e_2 = 0, \quad e_3 = \frac{1}{2}  \frac{\Delta
  z}{\sqrt{(\Delta x)^2 + (\Delta z)^2}},
\end{equation}
and the optimal element
\begin{equation}
  e_o = e_1 \sigma_1 + e_2 \sigma_2 + e_3 \sigma_3 = \frac{1}{2 \sqrt{(\Delta x)^2 + (\Delta z)^2}} \left( \begin{array}{cc}
    \Delta z & \Delta x\\
    \Delta x & - \Delta z
  \end{array} \right) .
\end{equation}
In this case, we have
\begin{equation}
  d (\rho_1, \rho_2) = |e_2 | \cdot | \Delta y| + \theta \sqrt{(\Delta x)^2 +
  (\Delta z)^2} = \frac{1}{2}  \sqrt{(\Delta x)^2 + (\Delta z)^2} .
\end{equation}
One can also express the Bloch vectors in spherical coordinates. Denote
$\Delta \vec{r} = \vec{r}_1 - \vec{r}_2 = r (\sin \theta \cos \phi, \sin
\theta \sin \phi, \cos \theta)$, where $\theta$ is the polar angle with
respect to the positive $y$-axis, $0 \leqslant \theta \leqslant \pi$, $\phi$
is the azimuthal angle in the $xz$-plane measured from the $x$-axis, $0
\leqslant \phi < 2 \pi$, and $r$ is the distance between the endpoints of the
vectors $\vec{r}_1$ and $\vec{r}_2$, $0 \leqslant r \leqslant 2$.

\begin{proposition}
  When $N = 2$, in the fuzzy torus $(A, H, D_1)$, the Connes spectral distance
  between two states $\rho_1$ and $\rho_2$ is
  \begin{equation}
    d (\rho_1, \rho_2) = \left\{ \begin{array}{ll}
      \dfrac{1}{2}  \sqrt{(\Delta x)^2 + (\Delta z)^2} = \dfrac{1}{2
      } r \sin \theta, & \qquad \dfrac{\pi}{4} \leqslant \theta
      \leqslant \dfrac{3 \pi}{4} ;\\[1em]
      \dfrac{| \Delta \vec{r} |^2}{4 | \Delta y|} = \dfrac{r}{4
      | \cos \theta |}, & \qquad \text{others.}
    \end{array} \right.
  \end{equation}
\end{proposition}

\begin{proof}
  The vector $\Delta \vec{r}$ satisfies $0 \leqslant | \Delta \vec{r} | = r
  \leqslant 2$, and we have
  \begin{equation}
    \sin \theta = \frac{\sqrt{(\Delta x)^2 + (\Delta z)^2}}{\sqrt{(\Delta x)^2
    + (\Delta y)^2 + (\Delta z)^2}} = \frac{\sqrt{(\Delta x)^2 + (\Delta
    z)^2}}{r}, \qquad \cos \theta = \frac{\Delta y}{r} .
  \end{equation}
  When $(\Delta x)^2 + (\Delta z)^2 \geqslant (\Delta y)^2$, we have
  $\dfrac{\sqrt{2}}{2} \leqslant \sin \theta \leqslant 1$, i.e.,
  $\dfrac{\pi}{4} \leqslant \theta \leqslant \dfrac{3 \pi}{4}$. Hence
  \begin{equation}
    d (\rho_1, \rho_2) = \frac{1}{2}  \sqrt{(\Delta x)^2 + (\Delta
    z)^2} = \frac{1}{2} r \sin \theta .
  \end{equation}
  When $(\Delta x)^2 + (\Delta z)^2 \leqslant (\Delta y)^2$, we have $0
  \leqslant \sin \theta \leqslant \dfrac{\sqrt{2}}{2}$, i.e., $0 \leqslant
  \theta \leqslant \dfrac{\pi}{4}$ or $\dfrac{3 \pi}{4} \leqslant \theta
  \leqslant \pi$. Hence
  \begin{equation}
    d (\rho_1, \rho_2) = \frac{| \Delta \vec{r} |^2}{4 | \Delta y|}
    = \dfrac{r}{4 | \cos \theta |} .
  \end{equation}
\end{proof}

It is easy to see that, in general, the above spectral distances are not invariant under unitary transformations of the states.

\section{Spectral Distance on the Twisted Fuzzy Torus}\label{sec5}

Since twisted fuzzy tori are important geometric objects in noncommutative
geometry, it is also significant to study Connes spectral distances on a
twisted fuzzy torus. Similar to the construction in Ref.~\cite{Dabrowski1}, to
obtain a simple but non-trivial conformally twisted model, one can consider the
projection with rank $r$,
\begin{equation}
  e_r = \mathrm{diag} (\underbrace{1, \ldots, 1}_r, \underbrace{0, \ldots,
  0}_{N - r}), \qquad 0 < r < N,
\end{equation}
and set
\begin{equation}
  k = \zeta [te_r + (1 - t) (1 - e_r)] \in A,
\end{equation}
where $\zeta > 0$, $0 \leqslant t \leqslant 1$.

For convenience, let us denote by $L_a, R_a$ the left and right multiplication
operators defined by $a \in A$, respectively. Using the above element $k$, one can define
the following Dirac operator with conformal rescaling
\begin{equation}
  k_J := JkJ^{- 1} = R_k, \qquad D_k := k_J Dk_J = R_k DR_k .
\end{equation}
Define
\begin{equation}
  \nu := L_{k^{- 1}} R_k = \mathrm{id}_{\mathbb{C}^4} \otimes
  \mathrm{Ad}_{k^{- 1}},
\end{equation}
where $\mathrm{Ad}_{k^{- 1}} (m) = k^{- 1} mk$, $m \in  \mathbb{M}_N (\mathbb{C})$. Then
$(A, H, D_k, J, \Gamma, \nu)$ forms an even $\nu$-twisted real spectral triple\cite{Brzezinski}.

\begin{lemma}\label{le7}
  For the operator norm, we have
  \begin{equation}
    \|A \otimes B\|_{op} = \|A\|_{op} \cdot \|B\|_{op} .
  \end{equation}
\end{lemma}

\begin{proof}
  \begin{equation}
    (A \otimes B)^{\dag}  (A \otimes B) = A^{\dag} A \otimes B^{\dag} B,
  \end{equation}
  Since the eigenvalues of $A^{\dag} A \otimes B^{\dag} B$ are the products of
  the eigenvalues of $A^{\dag} A$ and $B^{\dag} B$, and the operator norm
  $\|A\|_{op}$ is the square root of the largest eigenvalue of $A^{\dag} A$,
  we have
  \begin{equation}
    \|A \otimes B\|_{op} = \|A\|_{op} \cdot \|B\|_{op} .
  \end{equation}
\end{proof}

\begin{lemma}\label{le8}
  \cite{Magnus} Consider the matrices $P, Q \in \mathbb{M}_n (\mathbb{C})$, and define
  the operator $P_Q : \mathbb{M}_n (\mathbb{C}) \rightarrow \mathbb{M}_n
  (\mathbb{C})$,
  \begin{equation}
    P_Q (R) := PRQ, \quad R \in \mathbb{M}_n (\mathbb{C}),
  \end{equation}
  then we have
  \begin{equation}
    P_Q \cong Q^T \otimes P.
  \end{equation}
\end{lemma}

\begin{theorem}
  Assume $k \in A$ is a positive invertible element, and its inverse $k^{- 1}$
  is bounded. For any $a \in A$, there is
  \begin{equation}
    [D_k, L_a] = R_k  [D, L_a] R_k,
  \end{equation}
 and
  \begin{equation}
    \|[D_k, L_a]\|_{op} = \|k\|^2_{op} \cdot \|[D, L_a]\|_{op}.
  \end{equation}
The spectral distances of any states $\varphi, \psi$ satisfy the relation
  \begin{equation}
    d_{D_k}  (\varphi, \psi) = \|k\|_{op}^{- 2} d_D  (\varphi, \psi) .
  \end{equation}
\end{theorem}

\begin{proof}
  Since $L_a$ and $R_k$ commute, $L_a R_k = R_k L_a$, there is
  \begin{equation}\label{dka}
    [D_k, L_a] = R_k DR_k L_a - L_a R_k DR_k = R_k  (DL_a - L_a D) R_k = R_k 
    [D, L_a] R_k .
  \end{equation}
 Consider the action on some element $v \otimes m\in H$, there is
  \begin{eqnarray}
    R_k  [D, L_a] R_k  (v \otimes m) & = & [D, L_a]  (v \otimes (mk^2))
    \nonumber\\
    & = & [D, L_a]  (v \otimes m)  (\mathbb{I} \otimes k^2) . 
  \end{eqnarray}
  From Lemma \ref{le8}, one can obtain
  \begin{equation}
    R_k  [D, L_a] R_k \cong (\mathbb{I} \otimes k^2)^T \otimes [D, L_a],
  \end{equation}
  From the relation (\ref{dka}) and Lemma \ref{le7}, there are
  \begin{eqnarray}
    \|[D_k, L_a]\|_{op} & = & \|R_k [D, L_a] R_k \|_{op} \nonumber\\
    & = & \|(\mathbb{I} \otimes k^2)^T \|_{op} \cdot \|[D, L_a]\|_{op}
    \nonumber\\
    & = & \|k^2 \|_{op} \cdot \|[D, L_a]\|_{op} \nonumber\\
    & = & \|k\|^2_{op} \cdot \|[D, L_a]\|_{op} \nonumber\\
    & = & \|[\|k\|^2_{op} \cdot D, L_a]\|_{op} . 
  \end{eqnarray}
  Thus the ball condition of the Dirac operator $D_k$ is equivalent to that of
  $\|k\|^2_{op} \cdot D$. From the above relation (\ref{td}), the Connes spectral distance is,
  \begin{equation}
    d_{D_k}  (\varphi, \psi) = d_{\|k\|^2_{op} \cdot D}  (\varphi, \psi) =
    \|k\|_{op}^{- 2} \cdot d_D  (\varphi, \psi) .
  \end{equation}
  
\end{proof}

From the relations (\ref{dd12}) and (\ref{ij}), one can obtain the following result.
\begin{corollary}
  The spectral distance between the basis states $|i \rangle$ and $|j \rangle$
  in the $\nu$-twisted real spectral triple $(A, H, D_k, J, \Gamma, \nu)$ is
  \begin{equation}
    d_{D_k}  (|i \rangle, |j \rangle) = \frac{1}{ \lambda
      \|k\|_{o p}^2} \min \{|i -j |, N - |i - j|\}.
  \end{equation}
\end{corollary}

If $k = \zeta [te_r + (1 - t) (1 - e_r)]$ with $0 < t < 1$, then $\|k\|_{o p} =
\zeta \max \{t, 1 - t\}$, and thus
\begin{equation}
  d_{D_k}  (|i \rangle, |j \rangle) = \frac{1}{ \lambda
    \zeta^2 \max \{t, 1 - t\}^2} \min \{|i -j |, N - |i - j|\}.
\end{equation}

One can find that, in the $\nu$-twisted real spectral
triple $(A, H, D_k, J, \Gamma, \nu)$, the spectral distance depends on the conformal rescaling parameters $\zeta,t$, but not $r$.

\section{Discussions and conclusions}\label{sec6}

In this paper, we study the Connes spectral distance between states on the fuzzy torus. We construct a Dirac operator $D$ by commutators and anticommutators, and we find that the commutator part $D_1$ and
anticommutator part $D_2$ can be regarded as two mutually orthogonal components. Based on this Dirac operator, we construct a spectral triple of the fuzzy torus.

We study some properties of the spectral distance on the fuzzy torus. We find that the spectral triples $(A, H, D)$ and $(A, H, D_1)$, $(A, H, D_2)$ have similar metrics. There is a reciprocal Pythagorean theorem between the spectral distances of these spectral triples.

In order to study the Connes spectral distances between diagonal states, we construct a conditional expectation function of the optimal element which is just the corresponding diagonal part. This conditional expectation can lead to a contraction of the corresponding Lipschitz seminorm. Furthermore, we find that for any diagonal states, the corresponding optimal elements of spectral distances are also diagonal. In these diagonal cases, we only need to consider the ball condition in the diagonal entries of the optimal elements.

We explicitly calculate the spectral distances of some simple states, including basic states and some simple mixed states. We find that the spectral distances between basis states of
the fuzzy torus do not depend on the dimension $N$. There are some kinds of cyclic symmetry in both the optimal elements and the spectral distances between the diagonal states. But in general, these spectral distances are not invariant under unitary transformations of the states.
Furthermore, we also construct a fuzzy torus with some type of conformal twist, and study the relation between conformal parameters and spectral distances.

To our knowledge, these results have not been reported in literature.
These concrete examples are significant for studies of geometric structures of the fuzzy torus and other finite spectral triples.
They can help us to understand the mathematical relations of quantum states in the framework of noncommutative geometry. These methods and results can also be useful to
the studies of Connes spectral distances in other types of spectral triples.

\section*{Acknowledgments}
This work is partly supported by the
Guangdong Basic and Applied Basic Research Foundation (Grant No.
2024A1515010380), the 2024 Guangdong Province Education Science Planning
Project (Higher Education Special) (No. 2024GXJK455), the 2024 Guangdong
Higher Education Teaching Reform Project.

\section*{Appendix}
Let us prove Proposition \ref{pr2}.
Consider the spectral distance between a simple mixed state and the maximally
mixed state.
\begin{equation}
  \psi = \mathrm{diag} (p, 1 - p, 0, \ldots, 0), \quad \varphi = \frac{1}{N}
  \mathrm{diag} (1, 1, \ldots, 1), \quad 0 \leqslant p \leqslant 1.
\end{equation}
The spectral distance between $\psi$ and $\varphi$ is
\begin{eqnarray}
  d (\psi, \varphi) & = & \sup_{e \in B_o} \sum_{i = 1}^N (p_i - q_i) e_i
  \nonumber\\
  & = & \sup_{e \in B_o} \left[ \left( p - \frac{1}{N} \right) e_1 + \left( 1
  - p - \frac{1}{N} \right) e_2 - \frac{1}{N} e_3 - \cdots - \frac{1}{N} e_N
  \right] \nonumber\\
  & = & \sup_{e \in B_o} \left[ \left( p - \frac{1}{N} \right) \Delta_{12} +
  \frac{N - 2}{N} \Delta_{23} + \frac{N - 3}{N} \Delta_{34} + \cdots +
  \frac{1}{N} \Delta_{N - 1, N} \right],~~~~~~
\end{eqnarray}
where the real diagonal matrix $e = \mathrm{diag} (e_1, \ldots, e_N)$, and
$\Delta_{ij} = e_i - e_j$.

Obviously, there is
\begin{equation}
\frac{N - 2}{N} > \frac{N - 3}{N} > \cdots > \frac{1}{N} > 0.
\end{equation}
Using Theorem \ref{th3}, it is easy to see that, the above spectral distance function will achieve its
maximum when $\Delta_{i j} = \pm 1$ or $0$, and
\begin{equation}
\Delta_{12} + \Delta_{23} + \cdots + \Delta_{N - 1, N} = e_1 - e_N = 1 .
\end{equation}

When $N$ is even, $N = 2 m$, $m \geqslant 1$, if $p - \frac{1}{N} \leqslant
\frac{m - 1}{N}$, or $0 \leqslant p \leqslant \frac{1}{2}$, one can choose
\begin{eqnarray}
&&\Delta_{23} = \Delta_{34} = \cdots = \Delta_{m + 1, m + 2} =
   1, \nonumber\\
&& \Delta_{m + 2, m + 3} = \ldots = \Delta_{N - 1,
   N} = \Delta_{12} = - 1,
\end{eqnarray}
for example,
\begin{eqnarray}
&&  e_1 = m - 1,~~ e_2 = m,~~ e_3 =m - 1,~~ \ldots,~~ e_{m + 1} = 1,~~ e_{m + 2} = 0,
\nonumber\\
&&  e_{m + 3} = 1,~~ e_{m + 4} = 2,~~ \ldots,~~ e_{2
  m - 1} = m - 3,~~ e_N = m - 2,
\end{eqnarray}
and the spectral distance is
\begin{eqnarray}
  d (\psi, \varphi) & = &  \left[ \frac{N - 2}{N} +
  \frac{N - 3}{N} + \cdots + \frac{N - (m + 1)}{N} \right] \nonumber\\
  &&~~ - \left[ \frac{N -
  (m + 2)}{N} + \frac{N - (m + 3)}{N} + \cdots + \frac{1}{N} + \left( p -
  \frac{1}{N} \right) \right] \nonumber\\
  & = & \frac{N}{4} - p .
\end{eqnarray}
If $p - \frac{1}{N} \geqslant \frac{m - 1}{N}$, or $\frac{1}{2} \leqslant p
\leqslant 1$, then one can choose
\begin{equation}
\Delta_{12} = \Delta_{23} = \cdots = \Delta_{m, m + 1} =
   1, \quad \Delta_{m + 1, m + 2} = \ldots = \Delta_{N - 1,
   N} = -1,
\end{equation}
for example,
\begin{eqnarray}
&&  e_1 = m,~~ e_2 = m-1,~~ \ldots,~~ e_m =
  1,~~ e_{m + 1} = 0,
\nonumber\\
&&  e_{m + 2} = 1,~~ e_{m + 3} = 2,~~ \ldots,~~ e_{2
  m - 1} = m-2,~~ e_N = m-1,
\end{eqnarray}
and the spectral distance is
\begin{eqnarray}
  d (\psi, \varphi) & = & \left[ \left( p -
  \frac{1}{N} \right) + \frac{N - 2}{N} + \frac{N - 3}{N} + \cdots + \frac{N -
  m}{N} \right]\nonumber\\
    &&~~- \left[ \frac{N - (m + 1)}{N} + \frac{N - (m + 2)}{N} +
  \cdots + \frac{1}{N}  \right]\nonumber\\
  & = & \frac{N}{4} + p - 1 .
\end{eqnarray}

When $N$ is odd, $N = 2 m + 1$, if $p - \frac{1}{N} \leqslant \frac{N - (m +
2)}{N}$, or $0 \leqslant p \leqslant \frac{N - 1}{2 N}$, then
one can choose
\begin{eqnarray}
&&\Delta_{23} = \Delta_{34} = \cdots = \Delta_{m + 1, m + 2} =
   1, \quad \Delta_{m + 2, m + 3} = 0, \nonumber\\
&&   \Delta_{m + 3, m
   + 4} = \ldots = \Delta_{2 m, 2 m + 1} = \Delta_{12} = -1,
\end{eqnarray}
for example,
\begin{eqnarray}
&&  e_1 = m-1,~~ e_2 = m,~~ e_3 = m-1,~~ \ldots,~~ e_{m + 1} = 1,~~ e_{m + 2} = e_{m + 3}
  = 0,
\nonumber\\
&&
  e_{m + 4} = 1,~~ e_{m + 5} = 2, ~~\ldots,~~ e_{2
  m} = m-3,~~ e_{2 m + 1} = m-2,
\end{eqnarray}
and the spectral distance is
\begin{eqnarray}
  d (\psi, \varphi) & = &  \left[ \frac{N - 2}{N} +
  \frac{N - 3}{N} + \cdots + \frac{N - (m + 1)}{N} \right]\nonumber\\
    &&~~ - \left[ \frac{N -
  (m + 3)}{N} + \frac{N - (m + 4)}{N} + \cdots + \frac{1}{N} + \left( p -
  \frac{1}{N} \right) \right]\nonumber\\
  & = & \frac{N^2 - 1}{4 N} - p .
\end{eqnarray}
If $\frac{N - (m + 2)}{N} \leqslant p - \frac{1}{N} \leqslant \frac{N - (m +
1)}{N}$, or $\frac{N - 1}{2 N} \leqslant p \leqslant \frac{N + 1}{2 N}$, then
one can choose
\begin{eqnarray}
&&\Delta_{23} = \Delta_{34} = \cdots = \Delta_{m + 1, m + 2} =
   1, \quad \Delta_{12} = 0, \nonumber\\
&& \Delta_{m + 2, m + 3} =
   \Delta_{m + 3, m + 4} = \ldots = \Delta_{2 m, 2 m + 1} = -
   1,
\end{eqnarray}
for example,
\begin{eqnarray}
&&  e_1 = m,~~ e_2 = m,~~ e_3 = m-1, ~~\ldots,~~ e_{m + 1} = 1,~~ e_{m + 2} = 0,
\nonumber\\
&& 
  e_{m + 3} = 1,~~ e_{m + 4} = 2,~~ \ldots,~~ e_{2
  m} = m-2, ~~e_{2 m + 1} = m-1,
\end{eqnarray}
and the spectral distance is
\begin{eqnarray}
  d (\psi, \varphi) & = &  \left[  \frac{N - 2}{N} +
  \frac{N - 3}{N} + \cdots + \frac{N - (m + 1)}{N} \right]\nonumber\\
    &&~~ - \left[ \frac{N -
  (m + 2)}{N} + \frac{N - (m + 3)}{N} + \cdots + \frac{1}{N}  \right]\nonumber\\
  & = & \frac{(N - 1)^2}{4 N} .
\end{eqnarray}
If $\frac{N - (m + 1)}{N} \leqslant p - \frac{1}{N}$, or $\frac{N + 1}{2 N}
\leqslant p \leqslant 1$, then one can choose
\begin{eqnarray}
&&\Delta_{12} = \Delta_{23} = \Delta_{34} = \cdots \Delta_{m, m + 1} =
   1, \quad \Delta_{m + 1, m + 2} = 0, 
\nonumber\\
    &&\Delta_{m + 2, m
   + 3} = \Delta_{m + 3, m + 4} = \ldots = \Delta_{2 m, 2 m + 1} = -
   1,
\end{eqnarray}
for example,
\begin{eqnarray}
&&  e_1 = m,~~ e_2 = m-1,~~ \ldots,~~ e_m =
  1,~~ e_{m + 1} = e_{m + 2} = 0,
\nonumber\\
    &&
  e_{m + 3} = 1,~~ e_{m + 4} = 2,~~ \ldots, ~~e_{2
  m} = m - 2,~~ e_{2 m + 1} = m - 1,
\end{eqnarray}
and the spectral distance is
\begin{eqnarray}
  d (\psi, \varphi) & = & \left[ \left( p -
  \frac{1}{N} \right) + \frac{N - 2}{N} + \frac{N - 3}{N} + \cdots + \frac{N -
  m}{N} \right]\nonumber\\
    &&~~ - \left[ \frac{N - (m + 2)}{N} + \frac{N - (m + 3)}{N} +
  \cdots + \frac{1}{N}  \right]\nonumber\\
  & = & \frac{N^2 - 1}{4 N} - 1 + p .
\end{eqnarray}

\end{document}